\documentclass[11pt,a4paper]{article}

\usepackage{hyperref}	
\usepackage{array}
\usepackage{algorithm}
\usepackage{algpseudocode}
\usepackage{authblk}
\usepackage{bbm}
\usepackage[margin=1in]{geometry}
\usepackage[nodayofweek]{datetime}
\usepackage{amsfonts}
\usepackage{amsmath,amssymb}
\usepackage[amsmath,thmmarks,amsthm]{ntheorem}
\usepackage{mathtools}
\usepackage{enumitem}
\usepackage{cleveref}
\usepackage{microtype}
\usepackage{thm-restate}
\usepackage[normalem]{ulem}
\usepackage{pdfpages}
\usepackage{tikz}
\usepackage{xcolor}

\usetikzlibrary{backgrounds,calc,decorations.pathreplacing,shapes}

\setenumerate[1]{label={(\alph{enumi})}} 

\algnewcommand\algorithmicforeach{\textbf{for each}}
\algdef{S}[FOR]{ForEach}[1]{\algorithmicforeach\ #1\ \algorithmicdo}

\newcommand{\ADV}{\textit{ADV}}
\newcommand{\cost}{\textit{cost}}

\newcommand{\sint}{s\mathrlap{'}_A}
\newcommand{\OPT}{\textit{OPT}}

\newtheorem{theorem}{Theorem}

\newtheorem{observation}[theorem]{Observation}

\newtheorem{claim}[theorem]{Claim}
\newtheorem{corollary}[theorem]{Corollary}
\newtheorem{lemma}[theorem]{Lemma}

\newcommand{\vShort}[1]{\relax}
\newcommand{\vFull}[1]{#1}

\DeclarePairedDelimiter{\inte}{\lfloor}{\rceil}

\title{The Online $k$-Taxi Problem\thanks{Partially supported by the ERC Advanced Grant 321171 (ALGAME) and by EPSRC.}
}

\author{Christian Coester}
\author{Elias Koutsoupias}
\affil{Department of Computer Science, University of Oxford\\ \texttt{\{christian.coester,elias.koutsoupias\}@cs.ox.ac.uk}}

\begin{document}
	
\setcounter{page}{0}

\date{}
\maketitle

\begin{abstract}
We consider the online $k$-taxi problem, a generalization of the $k$-server problem, in which $k$ taxis serve a sequence of requests in a metric space. A request consists of two points $s$ and $t$, representing a passenger that wants to be carried by a taxi from $s$ to $t$. The goal is to serve all requests while minimizing the total distance traveled by all taxis. The problem comes in two flavors, called the easy and the hard $k$-taxi problem: In the easy $k$-taxi problem, the cost is defined as the total distance traveled by the taxis; in the hard $k$-taxi problem, the cost is only the distance of empty runs.

The hard $k$-taxi problem is substantially more difficult than the easy version with at least an exponential deterministic competitive ratio, $\Omega(2^k)$, admitting a reduction from the layered graph traversal problem. In contrast, the easy $k$-taxi problem has exactly the same competitive ratio as the $k$-server problem. We focus mainly on the hard version. For hierarchically separated trees (HSTs), we present a memoryless randomized algorithm with competitive ratio $2^k-1$ against adaptive online adversaries and provide two matching lower bounds: for arbitrary algorithms against adaptive adversaries and for memoryless algorithms against oblivious adversaries. Due to well-known HST embedding techniques, the algorithm implies a randomized $O(2^k\log n)$-competitive algorithm for arbitrary $n$-point metrics. This is the first competitive algorithm for the hard $k$-taxi problem for general finite metric spaces and general $k$. For the special case of $k=2$, we obtain a precise answer of $9$ for the competitive ratio in general metrics. With an algorithm based on growing, shrinking and shifting regions, we show that one can achieve a constant competitive ratio also for the hard $3$-taxi problem on the line (abstracting the scheduling of three elevators).
\end{abstract}
\clearpage

\section{Introduction}

The $k$-taxi problem, originally proposed by Karloff and introduced by Fiat et al. \cite{FiatRR90}, is a natural generalization of the fundamental $k$-server problem. In this problem, $k$ taxis are located in a metric space and need to serve a sequence $\sigma$ of requests. A request is a pair $(s,t)$ of two points in the metric space, representing a passenger that wants to travel from $s$ to $t$. A taxi serves the request by first moving to $s$ and then to $t$. The goal is to serve all requests in order while minimizing the total distance traveled by all taxis. We consider the online version of this problem where requests appear one by one, i.e., a new request is revealed only after the previous request has been served.

Since the distance from $s$ to $t$ needs to be traveled anyway --- independently of the algorithm's decisions --- it makes sense to exclude it from the cost and minimize only the overhead travel that actually depends on the algorithm choices. This is precisely what the hard $k$-taxi problem does: In the hard version, the cost is defined as the distance traveled while not carrying a passenger, i.e., the overhead distance traveled on top of the distances between the start-destination pairs. In contrast, for the easy $k$-taxi problem, the cost is the total distance traveled by the taxis. Thus, the cost of any taxi schedule differs by exactly the sum of the $s$-$t$-distances between the two versions, and in particular, the optimal offline solutions are the same for both versions. However, the different cost functions make it more difficult to approximate the optimal solution value of the hard version. Fiat et al.~\cite{FiatRR90} pointed out the two versions of this problem, and they were called easy taxicab problem and hard taxicab problem by Kosoresow~\cite{Kosoresow96}.

The problem was recently reintroduced as the Uber problem in \cite{DehghaniEHLS17}, which studied the easy version of the problem with the input being produced in a stochastic manner. Here we consider the adversarial case. This worst-case analysis is arguably useful for developing algorithms and improving our understanding of a problem, even when there is sufficiently large collected data to allow us to treat the problem as a stochastic one, which arguably is the case with Uber, Lyft etc.

Besides scheduling taxis, the $k$-taxi problem also models other tasks such as scheduling elevators (in which case the metric space is the line; see Section~\ref{sec:Line}) or transport vehicles in a factory, and other applications where people or objects need to be transported between locations.

\subsection{Previous results and related work}

The first competitive algorithm for the easy $k$-taxi problem was given by Fiat et al.~\cite{FiatRR90} when they introduced the problem, with a competitive ratio exponential in $k$. Following the finding of a $(2k-1)$-competitive algorithm for the $k$-server problem \cite{KoutsoupiasP95}, the competitive ratio of the easy $k$-taxi problem was improved to $2k+1$: Kosoresow~\cite{Kosoresow96} showed that if there is a $c$-competitive algorithm for the $k$-server problem, then there is a $(c+2)$-competitive algorithm for the easy $k$-taxi problem. This result was also established in \cite{DehghaniEHLS17} with a similar reduction.

For the hard $k$-taxi problem, Fiat et al.~\cite{FiatRR90} mentioned a competitive algorithm by Karloff for $k=2$. Based on Karloff's algorithm, Kosoresow~\cite{Kosoresow96} gave a $15$-competitive algorithm for $k=2$. No competitive algorithm is known for $k>2$.

The $k$-server problem \cite{ManasseMS88}, one of the most studied online problems, is the special case of the (easy and hard) $k$-taxi problem where for each request, start and destination are identical. Thus, the lower bound of $k$ on the competitive ratio of the $k$-server problem \cite{ManasseMS88} immediately implies the same lower bound for the $k$-taxi problem. According to the famous \emph{$k$-server conjecture}, this bound is tight for the $k$-server problem, yet the best known upper bound remains $2k-1$ \cite{KoutsoupiasP95}. The \emph{randomized $k$-server conjecture} states that a competitive ratio of $O(\log k)$ can be achieved by randomized algorithms on any metric space, and there has been tremendous progress on this question recently \cite{BansalBMN11,BubeckCLLM18,Lee18}. More information on the $k$-server problem is presented in \cite{Koutsoupias09}. Besides the $k$-taxi/Uber problem, recent work on other variants of the $k$-server problem include the $(h,k)$-server problem~\cite{BansalEJKP15,BansalEJK17}, the infinite server problem~\cite{CoesterKL17} and the weighted $k$-server problem~\cite{BansalEK17}.

The layered graph traversal problem, first introduced in~\cite{PapadimitriouY91}, is another deep problem in online computation and known to be equivalent to the metrical service systems problem \cite{FiatFKRRV98}. As we will see, the $k$-taxi problem generalizes not only the $k$-server problem but also the (deterministic) layered graph traversal problem. The best known bounds on the competitive ratio of layered graph traversal (for graphs of width $k$) are $\Omega(2^k)$ \cite{FiatFKRRV98} and $O(k2^k)$ \cite{Burley96} in the deterministic case, and $O(k^{13})$ and $\Omega\left(\frac{k^2}{\log^{1+\epsilon}(k)}\right)$ in the randomized case~\cite{Ramesh95}.

\subsection{Our results}
For the hard $k$-taxi problem, we show tight bounds on HSTs (defined in Section~\ref{subsec:prelim}) in two settings: for randomized algorithms against adaptive online adversaries, and for memoryless randomized algorithms against oblivious adveraries:
\begin{theorem}\label{thm:HSTs}
	There is a $(2^k-1)$-competitive memoryless randomized algorithm for the hard $k$-taxi problem on HSTs against adaptive online adversaries. This bound is tight in two senses: Any randomized algorithm $A$ for the hard $k$-taxi problem on HSTs has competitive ratio at least $2^k-1$ against adaptive online adversaries. If $A$ is memoryless, then its competitive ratio is at least $2^k-1$ even against oblivious adversaries.
\end{theorem}

Since the randomized competitive ratio against adaptive online adversaries is at most a square-root better than the deterministic competitive ratio for any online problem~\cite{Ben-DavidBKTW94}, this also implies a $4^k$-competitive deterministic algorithm for HSTs. More importantly, thanks to known probabilistic approximation of general $n$-point metrics by HSTs with distortion $O(\log n)$ \cite{FakcharoenpholRT03}, we obtain the first competitive algorithm for the hard $k$-taxi problem on general finite metrics:
\begin{corollary}\label{cor:general}
	There is an $O(2^k\log n)$-competitive randomized algorithm for the hard $k$-taxi problem on metric spaces of $n$ points.
\end{corollary}

For deterministic algorithms, we show a reduction from the layered width-$k$ graph traversal problem ($k$-LGT) to the hard $k$-taxi problem.
\begin{theorem}\label{thm:red}
	If there exists a $\rho$-competitive deterministic algorithm for the hard $k$-taxi problem, then there exists a $\rho$-competitive deterministic algorithm for $k$-LGT. In particular, the deterministic competitive ratio of the hard $k$-taxi problem is $\Omega(2^k)$.
\end{theorem}
This also shows that the hard $k$-taxi problem is substantially more difficult than the $k$-server problem.

Although the lower bound of $\Omega(2^k)$ also follows already from the first lower bound in Theorem~\ref{thm:HSTs}, the relation between the hard $k$-taxi problem and $k$-LGT is of independent interest, as it reveals an interesting connection between these two problems; a connection which can be potentially exploited to extend to lower bounds for randomized algorithms.

For the special case $k=2$, we improve the previous upper bound of $15$ \cite{Kosoresow96} and give tight bounds of $9$ on general metrics:
\begin{theorem}\label{thm:k=2}
	The deterministic competitive ratio of the hard $2$-taxi problem is exactly 9.
\end{theorem}
For the upper bound, we present a simple modification of the well-known Double Coverage algorithm with biased taxi speeds. The lower bound follows from the reduction from $k$-LGT and the fact that the deterministic competitive ratio of $2$-LGT is exactly $9$ \cite{PapadimitriouY91}.

With a significant extension of the algorithm, we obtain the following result for three taxis when the metric space is the real line:
\begin{theorem}\label{thm:3line}
There is an $O(1)$-competitive algorithm for the hard $3$-taxi problem on the line.
\end{theorem}
Our algorithm achieving this upper bound keeps track of regions around the taxis where they move more aggressively. These regions are continuously expanded, shrunk and shifted while the algorithm serves requests.

Lastly, we show the easy $k$-taxi problem to be exactly equivalent to the $k$-server problem, tightening the old result by removing the ``$+2$'' term:
\begin{restatable}{theorem}{easy} \label{thm:easy}
	The easy $k$-taxi problem has the same deterministic/randomized competitive ratio as the $k$-server problem. In particular, the deterministic competitive ratio of the easy $k$-taxi problem is between $k$ and $2k-1$, and it is $k$ if and only if the $k$-server conjecture holds.
\end{restatable}

\subsection{Preliminaries}\label{subsec:prelim}

Let $(M,d)$ be a metric space. A \emph{configuration} is a multiset of $k$ points in $M$, representing the positions of $k$ taxis. In the \emph{easy $k$-taxi problem} we are given an initial configuration and a sequence of requests $\sigma=(r_1,\dots,r_n)$ where $r_i=(s_i,t_i)\in M^2$. An algorithm must move the taxis so as to serve these requests in order. A taxi serves request $r_i$ by moving first to the start $s_i$ of the request and then to the destination $t_i$. An online algorithm has to make this decision without knowledge of the future requests. The cost is defined as the total distance traveled by all taxis.

The \emph{hard $k$-taxi problem} is defined in the same way except that the movement of the serving taxi from $s_i$ to $t_i$ is not counted towards the cost. Thus, the cost comprises only the overhead distance traveled while not carrying a passenger.

We write $A(C,\sigma)$ to refer to the run of algorithm $A$ on request sequence $\sigma$ starting from initial configuration $C$. The corresponding sequence of configurations is called a \emph{schedule}. By $\cost_A(C,\sigma)$ we denote its cost, although we will often omit the initial configuration from the notation and only write $\cost_A(\sigma)$. An algorithm is \emph{memoryless} if each decision depends only on the current configuration and the current request, but not the past configurations or requests.

The performance of an online algorithm is measured by comparing its cost to that of an adversary who generates a request sequence and also serves it. An algorithm $A$ is \emph{$\rho$-competitive} if $E(\cost_A(\sigma)) \le \rho E(\cost_\ADV(\sigma)) + c$ for every possible request sequence $\sigma$ generated by the adversary. Here, $\cost_A(\sigma)$ and $\cost_\ADV(\sigma)$ denote the cost of the algorithm and the adversary respectively to serve $\sigma$, and $c$ is a constant independent of $\sigma$. We consider two types of adversaries: The \emph{oblivious adversary} generates the request sequence in advance, independently of the outcome of random choices by the algorithm. Since this adversary knows in advance the sequence $\sigma$ it will generate, it can serve it \emph{offline} with optimal cost $\OPT(\sigma)$, i.e., $\cost_\ADV(\sigma)=\OPT(\sigma)$ for the oblivious adversary. The \emph{adaptive online adversary} makes the next request based on the algorithm's answers to previous requests, but serves it immediately. Thus, $\sigma$ is a random variable and $\cost_\ADV(\sigma)$ is not necessarily optimal. If no adversary type is stated explicitly, then we mean the oblivious adversary. For deterministic algorithms, the two notions are identical. The deterministic/randomized \emph{competitive ratio} of an online problem (and a given adversary type) is the infimum of all $\rho$ such that a deterministic/randomized $\rho$-competitive algorithm exists.

Clearly, it cannot be a disadvantage for the adversary to replace a request $(s,t)$ by two requests $(s,s)$ and $(s,t)$ because the adversary cost is unaffected by this and it forces the online algorithm to decide which taxi to send to $s$ before learning the destination $t$. For the request $(s,t)$, there is no decision to be made by the algorithm because it is clearly best to move the taxi already located at $s$ due to the previous request. Thus, we can assume without loss of generality that the adversary gives a request of the form $(s,t)$ with $s\ne t$ only if it is preceded by the request $(s,s)$. We call a requests of the form $(s,s)$ \emph{simple requests} and other requests \emph{relocation requests}. We may also say there is a request at $s$ to refer to a simple request $(s,s)$. The \emph{$k$-server problem} is the special case of the (easy and hard) $k$-taxi problem where all requests are simple, and in this case the taxis are called \emph{servers}. Conversely, the $k$-taxi problem is the same as the $k$-server problem except that the adversary can choose to relocate a pair of online and offline servers if they occupy the same point. In the hard $k$-taxi problem, relocation is free, and in the easy $k$-taxi problem, the algorithm and adversary both pay the distance of the relocation.

\paragraph{Hierarchically separated trees (HSTs) \cite{Bartal96}.} For $\alpha>1$, an \emph{$\alpha$-HST} is a tree where all leaves have the same combinatorial distance from the root and each node $u$ has a weight $w_u$ such that if $v$ is a child of $u$, then $w_u=\alpha w_v$. The metric space of an HST consists of its leaves only, and the distance between two leaves is defined as the weight of their least common ancestor. This is the shortest path metric when the distance from $u$ to its parent is defined as $\frac{\alpha-1}{2}w_u$ if $u$ is an internal vertex and $\frac{\alpha}{2}w_u$ if $u$ is a leaf. The significance of HSTs is that any metric space of $n$ points can be probabilistically embedded into a distribution over HSTs with distortion $O(\log n)$ \cite{FakcharoenpholRT03}.\footnote{The literature contains several slightly different definitions of HSTs; they all share the property of approximating arbitrary metrics with distortion $O(\log n)$.} Thus, a $\rho$-competitive algorithm for HSTs yields a randomized $O(\rho\log n)$-competitive algorithm for general metrics.

For nodes $x$ and $y$ of a tree, we denote by $P_{xy}$ their connecting path.

\subsection{Organization}
The hard $k$-taxi problem is studied in Section~\ref{sec:hard}. To prove Theorem~\ref{thm:HSTs}, we present the $(2^k-1)$-competitive algorithm for HSTs in Subsection~\ref{subsec:HST}, and the matching lower bounds in Subsections~\ref{sec:hstLb} and \ref{sec:LbMemoryless}. In Subsection~\ref{subsec:redLGT} we give the reduction from layered graph traversal that yields Theorem~\ref{thm:red} and the lower bound of Theorem~\ref{thm:k=2}. The upper bound of Theorem~\ref{thm:k=2} is shown in Subsection~\ref{subsec:k=2}. The algorithm for three taxis on the line (Theorem~\ref{thm:3line}) is presented in Subsection~\ref{sec:Line}. The equivalence of the easy $k$-taxi problem and the $k$-server problem (Theorem~\ref{thm:easy}) is shown in Section~\ref{sec:Easy}.

\vShort{Complete proofs are contained in the full version of this paper.}
\section{The hard \texorpdfstring{$k$}{k}-taxi problem}\label{sec:hard}

\subsection{An optimally competitive algorithm for HSTs}\label{subsec:HST}

We consider the following randomized algorithm, which we call \textsc{Flow}, for the $k$-taxi problem on HSTs. Suppose a simple request arrives at a leaf $s$ while the taxis are located at leaves $t_1,\dots,t_k$. For each taxi we need to specify its probability to serve $s$. Let $N$ be the Steiner tree of $s,t_1,\dots,t_k$, i.e., the minimum subtree of the HST that spans these leaves. We can think of $N$ as an electrical network by interpreting an edge of length $R$ as a resistor with resistance $R$. When sending a current of size $1$ through $N$ from source $s$ to sinks $t_1,\dots,t_k$, the resistances determine what fraction of the current flows into which sink. Algorithm \textsc{Flow} serves the request with a taxi from $t_i$ with probability equal to the fraction of current flowing into $t_i$.\footnote{Due to relocation requests, we may have several taxis at $t_i$. In this case, it does not matter which one we choose and what we describe is the combined probability of choosing one of them.} For a relocation request $(s,t)$ after a simple request $(s,s)$, \textsc{Flow} uses the taxi already located at $s$.\footnote{\textsc{Flow} is similar to the $k$-server algorithm RWALK by Coppersmith et al.~\cite{CoppersmithDRS93}: Given a weighted graph $(V,E)$ and interpreting edges as resistors as above, RWALK serves a request at $s$ with a server from $t$ with probability inversely proportional to the resistance of the resistor/edge $\{s,t\}$ (if the edge $\{s,t\}$ exists). RWALK is $k$-competitive for the metric of \emph{effective} resistances on $V$.}

To formalize this algorithm, we need to give a mathematical description of how much current flows into each sink. Let $\mathcal N$ be the set of subtrees $A$ of $N$ comprising at least one edge and with the property that, if $s_A$ is the (unique) node of $A$ closest to $s$ in $N$, then the leaves of $A$ are a subset of $s_A,t_1,\dots,t_k$. Formally, we view $A$ as the set of all nodes and edges of this subtree. We denote by $\kappa(A)=|A\cap\{t_1,\dots,t_k\}|$ the number of leaves of $A$ where a taxi is located. Note that $N\in\mathcal N$ and $\kappa(N)\le k$, with equality if and only if all taxis are located at different leaves. For each (sub)network $A\in\mathcal N$, we define (by induction on $\kappa(A)$) its resistance $R_A$ and describe what fraction of the current entering $A$ at $s_A$ flows to which sink in $A$. If $\kappa(A)=1$, then $A=P_{s_At_i}$ for some $i$. In this case, the resistance of $A$ is the length of this path, $R_A=d(s_A,t_i)$. Moreover, all current entering $A$ at $s_A$ flows to $t_i$.

If $\kappa(A)\ge 2$, then there is a unique node $\sint\in A$ such that $A=P_{s_A\sint}\cup B\cup C$ for some $B, C\in\mathcal N$ with $s_B=s_C=\sint$, and $P_{s_A\sint}$, $B\setminus\{\sint\}$ and $C\setminus\{\sint\}$ are disjoint. The resistance of $A$ is defined as
\begin{align}\label{eq:flowR}
R_A=d(s_A,\sint)+\frac{R_BR_C}{R_B+R_C}.
\end{align}
Current entering $A$ at $s_A$ flows entirely to $\sint$, where a $\frac{R_C}{R_B+R_C}$ fraction of it enters $B$ and the remaining $\frac{R_B}{R_B+R_C}$ fraction enters $C$.\footnote{It is easy to verify that $R_A$ and the current on each edge is well-defined, i.e.\ independent of the choice of $B$ and $C$. Note that if $s_A'$ has degree $\ge 2$ in $B$ or $C$, then the choice of $B$ and $C$ is not unique.}

Another interpretation of \textsc{Flow} is that we carry out a random walk from the request location to the taxi locations, and whichever taxi we end up at is the one that will serve the request. Whenever the random walk hits an intersection offering two possible directions to continue towards a taxi, either by entering a subtree $B$ or a subtree $C$, we choose the subtree with probability inversely proportional to its resistance.

The following theorem yields the upper bound of Theorem~\ref{thm:HSTs}. We do not actually need the HST property of geometrically decreasing weights, but only the weaker property that all requests are at the same distance from the root (in terms of the path metric extended to internal nodes).

\begin{theorem}
	\textsc{Flow} is $(2^k-1)$-competitive against adaptive online adversaries for the hard $k$-taxi problem in the leaf-space of any tree with uniform root-leaf-distances.
\end{theorem}
\begin{proof}
	We use a potential equal to $(2^k-1)$ times the value of a minimum matching $M$ of algorithm and adversary configurations. Since $M$ does not change upon relocation requests, we only need to consider simple requests. Whenever the adversary moves, the value of $M$ increases by at most the distance moved by the adversary. Thus, we only need to show for a simple request at a leaf $s$ where the adversary already has a taxi that
	\begin{align*}
	E(\cost)+(2^k-1)E(\Delta M) \le 0,
	\end{align*}
	where the random variables $\cost$ and $\Delta M$ denote the cost of \textsc{Flow} to serve the request and the associated increase of the minimum matching value.
	
	For a path $P$ we write $\ell(P)$ for its length.
	
	Let $t_s$ be the location of the \textsc{Flow} taxi matched to the adversary taxi at $s$ in $M$. Let $i$ be a random variable for the \textsc{Flow} taxi serving $s$, so that $t_i$ is its location before serving the request. We can partition the movement of taxi $i$ along the path $P_{t_is}$ into two parts, first the movement along $P_{t_is}\setminus P_{t_ss}$ and then the movement along $P_{t_i s}\cap P_{t_ss}$. During the first part, the value of $M$ can increase by at most $\ell(P_{t_is}\setminus P_{t_ss})$. Now, the value of $M$ is at most that of the matching $M'$ which differs from $M$ in that the adversary taxi at $s$ is matched to the \textsc{Flow} taxi $i$, and the \textsc{Flow} taxi at $t_s$ is matched to the adversary taxi previously matched to $i$. As taxi $i$ finishes its movement towards $s$, the value of $M'$ decreases by precisely $\ell(P_{t_is}\cap P_{t_ss})$. The value of the new minimum matching is bounded by the value of $M'$. Thus, we can bound the increase of the matching by
	\begin{align}\label{eq:flowMatchingBound}
	\Delta M\le \ell(P_{t_is}\setminus P_{t_ss})- \ell(P_{t_is}\cap P_{t_ss}).
	\end{align}
	On the right hand side, we simply count edges of $P_{t_is}$ negatively if they are also part of $P_{t_ss}$, and positively otherwise.
	
	For $A\in\mathcal N$ let
	\begin{align*}
	m(A) = E(\ell(A\cap P_{t_is}\setminus P_{t_ss})- \ell(A\cap P_{t_is}\cap P_{t_ss})\mid t_i\in A)
	\end{align*}
	be the expected contribution of edges from $A$ to the matching bound \eqref{eq:flowMatchingBound}, conditioned on \textsc{Flow} using a taxi that starts in $A$. Moreover, let
	\begin{align*}
	c(A) = E(\ell(A\cap P_{t_is})\mid t_i\in A)
	\end{align*}
	be the expected movement cost of \textsc{Flow} incurred on edges of $A$, conditioned on \textsc{Flow} using a taxi from $A$. For $A=P_{s_A\sint}\cup B\cup C\in\mathcal N$ as above, it follows from the definition of the algorithm that
	\begin{align}
	c(A)=d(s_A,\sint)+\frac{R_Cc(B)}{R_B+R_C}+\frac{R_Bc(C)}{R_B+R_C}.\label{eq:flowcA}
	\end{align}
	
	Since $E(\cost)=c(N)$ and $E(\Delta M)\le m(N)$, it suffices to show
	\begin{align}
	c(N)+(2^k-1)m(N) \le 0.\label{eq:flowPot2nd}
	\end{align}
	
	A key insight is provided by the following claim, relating $m(A)$ to $c(A)$ and $R_A$, which will allow us to reformulate \eqref{eq:flowPot2nd} purely in terms of the expected cost $c(N)$ and resistance $R_N$.
	
	\begin{claim*}
		For each $A\in\mathcal N$ with $t_s\in A$, $m(A)=c(A)-2R_A$.
	\end{claim*}
	\begin{proof}
		The proof is by induction on $\kappa(A)$. If $\kappa(A)=1$, then $A=P_{s_At_s}$; the condition $t_i\in A$ in the expectations defining $m(A)$ and $c(A)$ is equivalent to $t_i=t_s$. Thus, we have $m(A)=-\ell(P_{s_At_s})$ and $c(A)=\ell(P_{s_At_s})$. Since $R_A=d(s_A,t_s)=\ell(P_{s_At_s})$, the claim follows.
		
		If $\kappa(A)\ge 2$, then $A$ can be split into the path $P_{s_A\sint}$ and subtrees $B,C\in\mathcal N$ as above. Conditional on $t_i\in A$, the path $P_{s_A\sint}$ is contained in both $P_{t_is}$ and $P_{t_ss}$, so it contributes negatively to $m(A)$. Moreover, conditional on $t_i\in A$, the remaining part of $A\cap P_{t_is}$ is in $B$ with probability $\frac{R_C}{R_B+R_C}$ and in $C$ with probability $\frac{R_B}{R_B+R_C}$. Thus,
		\begin{align*}
		m(A)=-d(s_A,\sint)+\frac{R_Cm(B)}{R_B+R_C}+\frac{R_Bm(C)}{R_B+R_C}.
		\end{align*}
		Without loss of generality let $B$ be the subtree containing $t_s$. Then we can replace $m(B)$ in this formula by applying the induction hypothesis. Moreover, the edges of $P_{t_ss}$ do not intersect with $C$, and therefore $m(C)=c(C)$. This gives us
		\begin{align}
		m(A)=-d(s_A,\sint)+\frac{R_C(c(B)-2R_B)}{R_B+R_C}+\frac{R_Bc(C)}{R_B+R_C}.\label{eq:flowmA}
		\end{align}
		Combining \eqref{eq:flowcA} and \eqref{eq:flowmA}, we get
		\begin{align*}
		c(A)-m(A)=2d(s_A,\sint)+2\frac{R_BR_C}{R_B+R_C}=2R_A
		\end{align*}
		and the claim follows.
	\end{proof}
	Thanks to this claim, we can rewrite \eqref{eq:flowPot2nd} as
	\begin{align}
	2^{k-1} c(N)\le (2^k-1)R_N.\label{eq:flowcNRN}
	\end{align}
	Observe that unlike \eqref{eq:flowPot2nd}, the reformulation \eqref{eq:flowcNRN} no longer depends on the adversary configuration and the matching.
	
	To show \eqref{eq:flowcNRN}, we will need the following property: For $A\in\mathcal N$,
	\begin{align}\label{eq:flowRbound}
	\kappa(A)R_A\ge h(A),
	\end{align}
	where $h(A)$ denotes the minimal distance between $s_A$ and a taxi in $A$. This follows by an easy induction on $\kappa(A)$.
	
	We complete the proof by showing the following slightly stronger generalization of \eqref{eq:flowcNRN}:
	\begin{claim*}
		For $A\in\mathcal N$ with $s\in A$, we have $2^{\kappa(A)-1} c(A)\le (2^{\kappa(A)}-1)R_A$.
	\end{claim*}
	\begin{proof}
		We proceed again by induction on $\kappa(A)$. For $\kappa(A)=1$ the claim is easily seen to hold with equality.
		
		If $\kappa(A)\ge2$, let $B$ and $C$ be as before. Since $s\in A$, we have $s_A=s$.
		
		From \eqref{eq:flowR} and \eqref{eq:flowcA}, we get
		\begin{align}
		2^{\kappa(A)-1}&c(A) - (2^{\kappa(A)}-1)R_A\nonumber\\
		&= (1-2^{\kappa(A)-1})d(s,\sint) + 2^{\kappa(A)-1}\left(\frac{R_Cc(B)}{R_B+R_C}+\frac{R_Bc(C)}{R_B+R_C}\right) - (2^{\kappa(A)}-1)\frac{R_BR_C}{R_B+R_C}.\label{eq:flowLong}
		\end{align}
		We will show that this term is negative. Thus, for $\kappa(A)\ge 2$ the claim even holds with strict inequality.
		
		The simpler case is that $P_{s\sint}$ contains the parent node of $\sint$ in the HST. Then both $B$ and $C$ are contained in the subtree of the HST rooted at $\sint$; hence, all paths from $\sint$ to any of the taxis in $B$ or $C$ have length exactly\footnote{We are using here that an internal node of an HST is at the same distance from all its leaf descendants.} $h(B\cup C)$, and therefore $c(B)=c(C)=h(B\cup C)<d(s,\sint)$. Using this, as well as $\kappa(A)<2^{\kappa(A)}-1$ and applying \eqref{eq:flowRbound} to $B\cup C\in\mathcal N$, we get that term \eqref{eq:flowLong} is less than
		\begin{align*}
		(1-2^{\kappa(A)-1}&)d(s,\sint) + 2^{\kappa(A)-1}h(B\cup C) - \kappa(A)R_{B\cup C}\\
		&\le (1-2^{\kappa(A)-1})d(s,\sint) + (2^{\kappa(A)-1}-1)h(B\cup C)\\
		&< 0.
		\end{align*}
		
		In the other case, when $P_{s\sint}$ does not contain the parent of $\sint$ in the HST, we can still assume without loss of generality that also $B$ does not contain the parent of $\sint$. Then
		\begin{align}
		c(B)=h(B)=d(s,\sint),\label{eq:flowcB}
		\end{align}
		where the equality with $d(s,\sint)$ follows from the fact that the requested node $s$ is also a leaf in the subtree of the HST rooted at $\sint$.
		
		Since $P_{s\sint}\cup C\in\mathcal N$ and $\kappa(P_{s\sint}\cup C)=\kappa(C)<\kappa(A)$, we can apply the induction hypothesis to $P_{s\sint}\cup C$, yielding
		\begin{align*}
		2^{\kappa(C)-1} (d(s,\sint) +c(C))\le (2^{\kappa(C)}-1)(d(s,\sint)+R_C),
		\end{align*}
		and reordering,
		\begin{align}
		c(C)\le (1-2^{1-\kappa(C)})d(s,\sint)+(2-2^{1-\kappa(C)})R_C.\label{eq:flowInd}
		\end{align}
		Due to \eqref{eq:flowcB}, \eqref{eq:flowInd} and \eqref{eq:flowRbound}, we can bound term \eqref{eq:flowLong} by
		\begin{align*}
		(1-2^{\kappa(A)-1}&)h(B) + 2^{\kappa(A)-1}\frac{R_Ch(B)}{R_B+R_C} + (2^{\kappa(A)-1}-2^{\kappa(A)-\kappa(C)})\frac{R_Bh(B)}{R_B+R_C}
		\\&\qquad \qquad\qquad\qquad\qquad\qquad\qquad\qquad\qquad+ (1-2^{\kappa(A)-\kappa(C)})\frac{R_BR_C}{R_B+R_C}  \\
		&= h(B) - 2^{\kappa(B)}\frac{R_Bh(B)}{R_B+R_C} - (2^{\kappa(B)}-1)\frac{R_BR_C}{R_B+R_C}\\
		&< h(B) - \frac{R_Bh(B)}{R_B+R_C} - \kappa(B)\frac{R_BR_C}{R_B+R_C}\\
		&\le h(B) - \frac{R_Bh(B)}{R_B+R_C} - \frac{h(B)R_C}{R_B+R_C}\\
		&= 0
		\end{align*}
		and the claim follows.
	\end{proof}
	Invoking the claim for $A=N$, and using $\kappa(N)\le k$, we obtain \eqref{eq:flowcNRN}, concluding the proof of the theorem.
\end{proof}

\subsection{Lower bound against adaptive adversaries}\label{sec:hstLb}

We now show the first lower bound of Theorem~\ref{thm:HSTs}, matching the upper bound from the previous section.

Let $B_k^\alpha$ be the binary $\alpha$-HST of depth $k$ with vertex weights $\alpha^{k}, \alpha^{k-1},\dots,\alpha,1$ along root-to-leaf paths. For an infinite request sequence $\sigma$, we denote by $\sigma_{t}$ its prefix consisting of the first $t$ requests. The lower bound for adaptive adversaries in Theorem~\ref{thm:HSTs} follows from the following theorem by letting $\alpha\to\infty$ and $T\to\infty$:
\begin{theorem}\label{thm:LbBinHST}
	Let $\alpha\ge 3^k$. For each randomized algorithm $A$ for the hard $k$-taxi problem on $B_k^\alpha$, any fixed initial configuration, and any leaf $\ell$ of $B_k^\alpha$, one can construct online an infinite request sequence $\sigma$ such that
	\begin{enumerate}
		\item for each bounded stopping time $T$, there exists a deterministic online algorithm $\ADV$ (the adversary) such that\label{it:hstLbIneq}
		\begin{align*}E(\cost_A(\sigma_{T}))\ge \left(2^k-1-\frac{3^k}{\alpha}\right)E(\cost_\ADV(\sigma_{T}))-(2\alpha)^k,
		\end{align*}
		\item $\cost_A(\sigma_{t})\to\infty$ as $t\to\infty$ for all random choices of $A$, and\label{it:hstLbUnbounded}
		\item if the initial configuration is extended by adding a $(k+1)$st taxi at leaf $\ell$, then $\sigma$ can be served for free.\label{it:hstLbFree}
	\end{enumerate}
\end{theorem}
\begin{proof}
	We call an algorithm with an extra taxi as in \ref{it:hstLbFree} \emph{augmented algorithm}.
	
	We prove the theorem by induction on $k$. For $k=1$, $\sigma$ begins with a relocation request from the initial taxi position to the leaf of $B_1^\alpha$ other than $\ell$ and then places simple requests alternately at the two leaves of $B_1^\alpha$. The adversary follows the unique strategy to serve the requests.
	
	For the induction step, suppose the theorem holds for $k$ and we want to show it for $k+1$. The infinite request sequence $\sigma$ consists of several phases. Note that $B_{k+1}^\alpha$ contains two copies of $B_k^\alpha$ as subtrees. We call one of these two subtrees \emph{active} and the other one \emph{passive}, and these roles change after each phase. We also call a taxi \emph{active} or \emph{passive} if it is in the according subtree. All requests of a phase are in the active subtree, except for some relocation requests at the end of a phase. In phase $1$, the subtree containing $\ell$ is active, and we let $\ell_1=\ell$. We maintain the invariant that at the beginning of phase $i$, $k$ taxis are active and one is passive, and we denote the leaf in the passive subtree where the latter is located by $\ell_{i+1}$. As the active and passive subtree change their roles after each phase, $\ell_i$ is always in the active subtree of phase $i$. Clearly the invariant can be ensured for phase $1$ by relocation requests in the beginning.
	
	We define now the requests of phase $i$. As long as $A$ does not activate its passive taxi (i.e.\ move it from $\ell_{i+1}$ to the active subtree), we can interpret the behaviour of $A$ as that of an algorithm $A_i$ for the $k$-taxi problem in the active subtree. To make $A_i$ a full-fledged $k$-taxi algorithm for $B_k^\alpha$, we define it arbitrarily from the point when $A$ activates the passive taxi onwards. By the induction hypothesis, we can construct online a request sequence $\sigma^i$ in the active subtree such that for any bounded stopping time $T_i$ there is an adversary $\ADV_i$ with
	\begin{align}\label{eq:hstLbIndHyp}
	E(\cost_{A_i}(\sigma^i_{T_i})\mid \mathcal F_i)\ge \left(2^k-1-\frac{3^k}{\alpha}\right)E(\cost_{\ADV_i}(\sigma^i_{T_i})\mid \mathcal F_i)-(2\alpha)^k,
	\end{align}
	where $\mathcal F_i$ contains the information about $A$'s decisions before phase $i$.
	Moreover, $\cost_{A_i}(\sigma^i_{t})\to\infty$ as $t\to\infty$, and an augmented algorithm with an extra taxi at $\ell_i$ serves $\sigma^i$ for free.
	
	Phase $i$ commences with the requests of $\sigma^i$ until $A$ activates its passive taxi. If $A$ never activates the passive taxi, then phase $i$ never ends. Otherwise, once $A$ activates the passive taxi, we use relocation requests to move $k$ taxis from the active to the passive subtree (i.e.\ the active subtree of the next phase), which ends phase $i$ and satisfies the aforementioned invariant for phase $i+1$. The $k$ starting points of these relocation requests are the final taxi positions of $\ADV_i$.
	
	We need to show that this sequence satisfies the claimed properties.
	
	We first show that the corresponding statement of \ref{it:hstLbIneq} for $k+1$ instead of $k$ holds. Let $n(T)$ be the number of phases of the request sequence $\sigma_{T}$. Note that all but possibly the last phase include an activation. For $i\le n(T)$, let $T_i$ be the number of requests in phase $i$ until $A$ activates the passive taxi or (for $i=n(T)$) until time $T$ is reached. For $i>n(T)$ let $T_i=0$. Note that \eqref{eq:hstLbIndHyp} holds even for $i>n(T)$ (with $A_i$ and $\sigma^i$ defined arbitrarily if $\sigma$ has less than $i$ phases) and we may multiply the term $(2\alpha)^k$ by $\mathbbm{1}_{\{T_i>0\}}=\mathbbm{1}_{\{i\le n(T)\}}$. Since activating the passive taxi is at least $\alpha^k(\alpha-1)$ more expensive than using an already active taxi, the cost of $A$ in phase $i$ is at least $\cost_{A_i}(\sigma^i_{T_i})+\mathbbm{1}_{\{i<n(T)\}}\alpha^k(\alpha-1)$. Thus,
	\begin{align}
	E(&\cost_A(\sigma_T))\nonumber\\
	&\ge E\left(\sum_{i=1}^{\infty} \cost_{A_i}(\sigma^i_{T_i}) + \mathbbm{1}_{\{i<n(T)\}}\alpha^k(\alpha-1)\right)\nonumber\\
	&= \left(\sum_{i=1}^{\infty} E(\cost_{A_i}(\sigma^i_{T_i}))\right) + (E(n(T))-1)\alpha^k(\alpha-1)\nonumber\\
	&\ge \left(\sum_{i=1}^{\infty} \left(2^k-1-\frac{3^k}{\alpha}\right)E(\cost_{\ADV_i}(\sigma^i_{T_i}))-P(i\le n(T))(2\alpha)^k\right) + E(n(T))\alpha^k(\alpha-1) - \alpha^{k+1}\nonumber\\
	&\ge \left(2^{k}-1-\frac{3^k}{\alpha}\right)E\left(\sum_{i=1}^{n(T)} \cost_{\ADV_i}(\sigma^i_{T_i})\right) +E(n(T))\alpha^k(\alpha-3^{k}) - \alpha^{k+1}.\label{eq:hstLbCost}
	\end{align}
	The assumption that $T$ is bounded guarantees that all sums have finitely many non-zero summands.
	
	We will consider three different strategies for the adversary: The first strategy is to activate the passive offline taxi by moving it to $\ell_i$ at the start of each phase $i$. Since the other $k$ adversary taxis cover the active online taxis at the start of the phase, this allows the adversary to serve the requests of the phase for free (thanks to \ref{it:hstLbFree}). Activating the passive taxi costs $\alpha^{k+1}$, so this strategy incurs a cost of
	\begin{align}\label{eq:hstLbSwitch}
	n(T)\alpha^{k+1}.
	\end{align}
	In the second strategy we consider, the adversary never moves a taxi from one subtree $B_{k}^\alpha$ to the other, except when this happens due to a relocation request. In this case, phase $i$ is free only for even $i$. If $i$ is odd, it copies the behaviour of $\ADV_i$ to serve the requests of phase $i$. The cost of this strategy is
	\begin{align}\label{eq:hstLbOdd}
	\sum_{\stackrel{i=1}{i\text{ odd}}}^{n(T)}\cost_{\ADV_i}(\sigma^i_{T_i}).
	\end{align}
	The third strategy is similar, but the offline algorithm moves the passive taxi from $\ell_2$ to $\ell_1$ before phase $1$, for a cost of $\alpha^{k+1}$. Analogously to the second strategy, this strategy incurs cost
	\begin{align}\label{eq:hstLbEven}
	\alpha^{k+1}+\sum_{\stackrel{i=2}{i\text{ even}}}^{n(T)}\cost_{\ADV_i}(\sigma^i_{T_i}).
	\end{align}
	The actual strategy $\ADV$ is the one of these three which has the smallest cost in expectation.
	
	Since $E(\cost_{\ADV}(\sigma_T))$ is bounded by the expectations of \eqref{eq:hstLbOdd} and \eqref{eq:hstLbEven}, we have
	\begin{align*}
	2 E(\cost_{\ADV}(\sigma_T))\le \alpha^{k+1}+E\left(\sum_{i=1}^{n(T)}\cost_{\ADV_i}(\sigma^i_{T_i})\right).
	\end{align*}
	Therefore, with \eqref{eq:hstLbCost} we get
	\begin{align}
	E(\cost_A&(\sigma_T))\ge \left(2^{k+1}-2-\frac{2\cdot3^k}{\alpha}\right)E(\cost_{\ADV}(\sigma_T)) +E(n(T))\alpha^k(\alpha-3^k) - (2\alpha)^{k+1}.\label{eq:hstLbIntermediate}
	\end{align}
	Since $E(\cost_{\ADV}(\sigma_T))$ is bounded by the expectation of \eqref{eq:hstLbSwitch}, for $\alpha\ge 3^k$ we can bound the middle term of \eqref{eq:hstLbIntermediate} by
	\begin{align*}
	E(n(T))\alpha^k(\alpha-3^k) &\ge \frac{E(n(T))\alpha^k(\alpha-3^k)}{E(n(T))\alpha^{k+1}}E(\cost_{\ADV}(\sigma_T))\\
	&= \left(1 - \frac{3^k}{\alpha}\right)E(\cost_{\ADV}(\sigma_T)).
	\end{align*}
	Therefore,
	\begin{align*}
	E(\cost_A(\sigma_T)) &\ge \left(2^{k+1}-1-\frac{3^{k+1}}{\alpha}\right)E(\cost_{\ADV}(\sigma_T)) - (2\alpha)^{k+1}.
	\end{align*}
	
	The induction step of \ref{it:hstLbUnbounded} is fairly straight-forward: If the number of phases $n(t)\to\infty$ as $t\to\infty$, then it follows from the fact that the cost increases by at least $\alpha^{k+1}$ per phase as this is the cost of activating the passive taxi. Otherwise, the length of the last phase goes to infinity, and so does $A$'s cost during this phase by definition of the phases.
	
	For \ref{it:hstLbFree}, we show by induction on $i$ that an augmented algorithm with extra taxi starting at $\ell$ can serve all requests before phase $i$ for free, and at the start of phase $i$ it ends up in the configuration of $A$ with an extra taxi at $\ell_i$. For $i=1$ this is obvious by choice $\ell_1=\ell$. Suppose now this holds for some $i$. Since the augmented algorithm has an extra taxi at $\ell_i$ at the beginning of phase $i$, this phase is free as well by choice of the request sequence. Thus, all requests before phase $i+1$ are free. Moreover, since the requests of phase $i$ are free, this means that all requests of phase $i$ are relocation requests or simple requests at points where the augmented algorithm has a taxi at that time. But then it follows that the configuration of the augmented algorithm is always that of $A$ with an extra taxi. Since $A$ served the last simple request of phase $i$ by moving a taxi away from $\ell_{i+1}$, the position of the extra taxi is $\ell_{i+1}$ when phase $i+1$ starts.
\end{proof}

\subsection{Lower bound for memoryless algorithms}\label{sec:LbMemoryless}
We now show the other lower bound of Theorem~\ref{thm:HSTs}. For a configuration $C$ and a request sequence $\sigma$, we denote by $w(C,\sigma)$ the optimal cost of a schedule that, starting from configuration $C$, serves $\sigma$ and then returns to configuration $C$. We will need the following lemma.

\begin{lemma}\label{lem:LbMemorylessFree}
Let $C_0',\dots,C_{n}'$ be a $(k+1)$-taxi schedule of cost $0$ for a request sequence $\sigma$.
\begin{enumerate}
    \item If $C_0,\dots,C_{n}$ is a $k$-taxi schedule for $\sigma$ with $C_0\subset C_0'$, then $C_i\subset C_i'$ for all $i$.\label{it:LemLbMemorylessFreek}
    \item If $C_0,\dots,C_n$ is a $(k+1)$-taxi schedule for $\sigma$ with $C_0\setminus C_0'=\{\ell\}$ for some $\ell\notin C_n$, then $C_n= C_n'$.\label{it:LemLbMemorylessFreek+1}
\end{enumerate}
\end{lemma}
\begin{proof}
To prove \ref{it:LemLbMemorylessFreek}, we proceed by induction on $i$. If the $k$-taxi algorithm incurs no cost when passing from $C_{i-1}$ to $C_{i}$ to serve the $i$th request $r_i$, then either $C_{i-1}=C_i$ and $r_i$ is a simple request at a point of $C_i$, or $r_i$ is a relocation request $(s_i,t_i)$ with $s_i\in C_{i-1}$ and $C_i=C_{i-1}\setminus\{s_i\}\cup\{t_i\}$. In both cases, $C_i\subset C_i'$ follows from $C_{i-1}\subset C_{i-1}'$.

Otherwise, without loss of generality $r_i$ is a simple request at some $t_i\notin C_{i-1}$, and $C_i=C_{i-1}\setminus\{x_i\}\cup\{t_i\}$ for some $x_i\in C_{i-1}$. Since the schedule $C_0',\dots,C_n'$ has cost $0$, we must have $C_{i-1}'=C_i'=C_{i-1}\cup\{t_i\}=C_i\cup\{x_i\}$.

Part \ref{it:LemLbMemorylessFreek+1} follows from part \ref{it:LemLbMemorylessFreek} if we replace $C_i'$ by $C_i'\cup\{\ell\}$ and $k$ by $k+1$.
\end{proof}

The lower bound of Theorem~\ref{thm:HSTs} for memoryless algorithms against oblivious adversaries follows from the following Theorem by letting $N\to\infty$. We use again the binary $\alpha$-HSTs $B_k^\alpha$ from the previous section.

\begin{theorem}\label{thm:LbMemoryless}
For $N\in\mathbb N$ sufficiently large, $\alpha=N^2$, each memoryless algorithm $A_k$ for the hard $k$-taxi problem on $B_k^\alpha$ and any leaf $\ell\in B_k^\alpha$, there exists a configuration $C$ of $k$ distinct points and a request sequence $\sigma$ such that
\begin{enumerate}
	\item $E(\cost_{A_k}(C,\sigma)) \ge (2^k-1)(2\alpha)^k(N-2^k)$,\label{it:LbMemorylessAlg}
	\item $0<w(C,\sigma)\le (2\alpha)^k (N + k)$,\label{it:LbMemorylessOpt}
	\item $w(C\cup\{\ell\},\sigma)=0$.\label{it:LbMemorylessFree}
\end{enumerate}
\end{theorem}
\begin{proof}
Let us first fix some notation. For request sequences $\sigma_1$ and $\sigma_2$, we denote by $\sigma_1\sigma_2$ their concatenation. For $m\ge 1$, let $\sigma_1^m$  be the $m$-fold repetition of $\sigma_1$, i.e., $\sigma_1^0$ is the empty sequence and $\sigma_1^{m+1}=\sigma_1^m\sigma_1$. For sets of points $X=\{x_1,\dots,x_n\}$ and $Y=\{y_1,\dots,y_n\}$ with $x_1,\dots,x_n$ and $y_1,\dots,y_n$ sorted lexicographically, we denote by $(X\to Y)$ the sequence of relocation requests $(x_1,y_1)\dots(x_n,y_n)$. Thus, if an algorithm is in configuration $X$, then the request sequence $(X\to Y)$ changes the configuration to $Y$. Abusing notation, we also write $X$ for the sequence of simple requests $x_1\dots x_n$. So if $X$ is a configuration of distinct points, then the request sequence $X^m$ forces an algorithm to either move to the configuration $X$ or suffer large cost (if $m$ is large). For an algorithm $A$, an initial configuration $C$ and a request sequence $\sigma$, we denote by $A(C,\sigma)$ the corresponding schedule.

We prove the theorem by induction on $k$. For $k=1$, we can write $B_1^\alpha=\{\ell,r\}$ for some $r$. The theorem holds for $C=\{r\}$ and $\sigma=(\ell r)^N$.

For the induction step, suppose the theorem holds for some fixed $k$ and we want to prove it for $k+1$. Say we are given an algorithm $A_{k+1}$ for the $k$-taxi problem and a leaf $\ell\in B_{k+1}^\alpha$. We will refer to the subtree $B_k^\alpha$ containing $\ell$ as the \emph{left subtree} and the other subtree $B_k^\alpha$ as the \emph{right subtree}. For a leaf $r$ in the right subtree, write $A_{k+1}|_r$ for the $k$-taxi algorithm in the left subtree that behaves like $A_{k+1}$ conditioned on having one taxi at $r$ that does not move.\footnote{If for some configuration $C\cup\{r\}$, $A_{k+1}$ moves the taxi from $r$ with probability $1$ for a given request, the move of $A_{k+1}|_r$ from $C$ is defined arbitrarily.}
Let $C_{r\ell}$ and $\sigma_{r\ell}$ be the $k$-taxi configuration and request sequence in the left subtree induced by the induction hypothesis applied to $A_{k+1}|_r$ and $\ell$. For a request sequence $\sigma'$, let $F_{r}(\sigma')$ denote the final configuration of $A_{k+1}(C_{r\ell}\cup\{r\},\sigma')$.

In the following, $m$ will be some large integer. Let
\begin{align*}
p_{r\ell,m} &= P(r\notin F_r(\sigma_{r\ell}C_{r\ell}^m))\\
\epsilon_{r\ell,m} &= P(C_{r\ell}\not\subset F_r(\sigma_{r\ell}C_{r\ell}^m)).
\end{align*}
Note that $p_{r\ell,m}$ is a non-decreasing and $\epsilon_{r\ell,m}$ a non-increasing function of $m$. Thus, we can define
\begin{align*}
p_{r\ell}&=\lim_{m\to\infty} p_{r\ell,m}\\
\epsilon_{r\ell}&=\lim_{m\to\infty} \epsilon_{r\ell,m}.
\end{align*}
Define $C_{\ell r}$, $\sigma_{\ell r}$, $p_{\ell r,m}$, $\epsilon_{\ell r,m}$, $p_{\ell r}$, $\epsilon_{\ell r}$ and $F_\ell$ similarly with the roles of $\ell$ and $r$ reversed.

Roughly speaking, the values $p_{r\ell}$ and $p_{\ell r}$ indicate how aggressively the algorithm moves between the two subtrees. We use two different definitions for $\sigma$ depending on whether these values are large or small. In both cases, we set $C=C_{r\ell}\cup\{r\}$.

\ \\
\noindent \textbf{Case 1:} $p_{\ell r}+p_{r\ell}> \frac{2^k}{N}$ and $p_{r\ell}>0$ and $p_{\ell r}>0$.

\noindent The central building blocks of $\sigma$ are the subsequences
\begin{align*}
\sigma_r = \sigma_{r\ell}C_{r\ell}^m(C_{r\ell}\to C_{\ell r})(\sigma_{\ell r}C_{\ell r}^m)^m(C_{\ell r}\to C_{r\ell})
\end{align*}
and $\sigma_\ell$ defined in the same way with $\ell$ and $r$ reversed. The idea of $\sigma_r$ is that when starting from configuration $C_{r\ell}\cup\{r\}$, the prefix $\sigma_{r\ell}C_{r\ell}^m$ shall lure the algorithm to move the taxi from $r$ to the left subtree. If it does so, it will be punished during the part $(\sigma_{\ell r}C_{\ell r}^m)^m$, which forces all taxis back to the right subtree. Since $p_{\ell r}+p_{r\ell}> \frac{2^k}{N}$, at least one of $\sigma_r$ and $\sigma_\ell$ will successfully exploit the algorithm's aggressiveness. We define the entire request sequence as
\begin{align*}
\sigma = \sigma_r^\alpha(C_{r\ell}\to C_{\ell r})(C_{\ell r}\cup\{\ell\})^m\sigma_\ell^\alpha(C_{\ell r}\to C_{r\ell}).
\end{align*}

We first prove the induction step of \ref{it:LbMemorylessOpt}. Clearly, $0<w(C,\sigma)$. Moreover, by the induction hypothesis, we have $w(C_{r\ell}\cup\{r\},\sigma_{r\ell})\le w(C_{r\ell},\sigma_{r\ell})\le (2\alpha)^k (N + k)$ and $w(C_{\ell r}\cup\{r\},\sigma_{\ell r})=0$. Therefore, $w(C_{r\ell}\cup\{r\},\sigma_r)\le (2\alpha)^k (N + k)$. By symmetry, we also have $w(C_{\ell r}\cup\{\ell\},\sigma_\ell)\le (2\alpha)^k(N + k)$. Thus,
\begin{align*}
w(C,\sigma) &\le \alpha w(C_{r\ell}\cup\{r\},\sigma_r) + \alpha^{k+1} + \alpha w(C_{\ell r}\cup\{\ell\},\sigma_\ell) + \alpha^{k+1}\\
&\le (2\alpha)^{k+1}(N+k+1).
\end{align*}

To see \ref{it:LbMemorylessFree}, it follows easily from the induction hypothesis that, starting from $C\cup\{\ell\}= C_{r\ell}\cup\{\ell,r\}$, serving $\sigma$ incurs no cost and makes the algorithm return to $C\cup\{\ell\}$ in the end.

The most technical part of this proof is the induction step of \ref{it:LbMemorylessAlg}. We can assume that
\begin{align}
    \limsup_{m\to\infty}E(\cost_{A_{k+1}}(C,\sigma))<\infty,\label{eq:LbMemorylessAss}
\end{align}
since otherwise \ref{it:LbMemorylessAlg} follows immediately for some large choice of $m$.

Let us first examine the behaviour of $A_{k+1}(C_{r\ell}\cup\{r\}, \sigma_{r\ell}C_{r\ell}^m)$. With probability $p_{r\ell,m}$, the algorithm moves the taxi from $r$ to the left subtree for cost $\alpha^{k+1}$. Otherwise (with probability $1-p_{r\ell,m}$), the taxi at $r$ stays put; conditioned on this, the expected cost suffered during the prefix $\sigma_{r\ell}$ is at least $(2^k-1)(2\alpha)^k(N-2^k)$ by the induction hypothesis. Moreover, if $F_r(\sigma_{r\ell}C_{r\ell}^m)\not\supset C_{r\ell}$, then $F_r(\sigma_{r\ell}\sigma')\not\supset C_{r\ell}$ for any prefix $\sigma$ of $C_{r\ell}^m$ and the algorithm incurs cost at least $\alpha$ during each of the $m$ subsequences $C_{r\ell}$.
Overall, we have
\begin{align}\label{eq:LbMemoryless1}
E(\cost_{A_{k+1}}(C_{r\ell}\cup\{r\}, \sigma_{r\ell}C_{r\ell}^m)) \ge p_{r\ell,m}\alpha^{k+1} + (1-p_{r\ell,m})(2^k-1)(2\alpha)^k(N-2^k) + \epsilon_{r\ell,m}\alpha m.
\end{align}

\begin{claim}\label{cl:LbMemorylessIff}
$r\notin F_r(\sigma_{r\ell}C_{r\ell}^m)$ if and only if $F_r(\sigma_{r\ell}C_{r\ell}^m)=C_{r\ell}\cup\{\ell\}$.
\end{claim}
\begin{proof}
The direction ``if'' holds since $C_{r\ell}\cup\{\ell\}$ is contained in the left subtree while $r$ is in the right subtree. ``Only if'' follows from part \ref{it:LbMemorylessFree} of the induction hypothesis and Lemma~\ref{lem:LbMemorylessFree}\ref{it:LemLbMemorylessFreek+1}.
\end{proof}

From the claim it also follows that the two events defining $\epsilon_{r\ell,m}$ and $p_{r\ell,m}$ are disjoint. So with probability $1-p_{r\ell,m}-\epsilon_{r\ell,m}$, none of the two events happens. This implies
\begin{align}
&P(F_{r}(\sigma_{r\ell}C_{r\ell}^m)=C_{r\ell}\cup\{\ell\})=p_{r\ell,m}\label{eq:LbMemorylessp}\\
&P(F_{r}(\sigma_{r\ell}C_{r\ell}^m)=C_{r\ell}\cup\{r\}) = 1-p_{r\ell,m} - \epsilon_{r\ell,m}\label{eq:LbMemorylesspeps}.
\end{align}
For $i\ge 0$, we have
\begin{align}
    &P(F_r(\sigma_{r\ell}C_{r\ell}^m(C_{r\ell}\to C_{\ell r})(\sigma_{\ell r}C_{\ell r}^m)^i)=C_{\ell r}\cup\{\ell\}) \nonumber\\
    &\ge P(F_r(\sigma_{r\ell}C_{r\ell}^m)=C_{r\ell}\cup\{\ell\} \land \forall j=1,\dots, i\colon F_r(\sigma_{r\ell}C_{r\ell}^m(C_{r\ell}\to C_{\ell r})(\sigma_{\ell r}C_{\ell r}^m)^j)=C_{\ell r}\cup\{\ell\}) \nonumber\\
    &= P(F_r(\sigma_{r\ell}C_{r\ell}^m)=C_{r\ell}\cup\{\ell\}) P(F_\ell(\sigma_{\ell r}C_{\ell r}^m)=C_{\ell r}\cup\{\ell\})^i \nonumber\\
    &= p_{r\ell,m}(1-p_{\ell r,m}-\epsilon_{\ell r,m})^i,\label{eq:LbMemorylessppeps}
\end{align}
where the first equation uses memorylessness of $A_{k+1}$ and the last equation uses \eqref{eq:LbMemorylessp} and the symmetric version of \eqref{eq:LbMemorylesspeps}.

Using \eqref{eq:LbMemorylessppeps}, \eqref{eq:LbMemoryless1} and the symmetric version of \eqref{eq:LbMemoryless1}, we can bound the cost during $\sigma_r$ by
\begin{align*}
&E(\cost_{A_{k+1}}(C_{r\ell}\cup\{r\},\sigma_r))\\
&\ge E(\cost_{A_{k+1}}(C_{r\ell}\cup\{r\},\sigma_{r\ell}C_{r\ell}^m) + \sum_{i=0}^{m-1}p_{r\ell,m}(1-p_{\ell r,m}-\epsilon_{\ell r,m})^i E(\cost_{A_{k+1}}(C_{\ell r}\cup\{\ell\},\sigma_{\ell r}C_{\ell r}^m)\\
&\ge p_{r\ell,m}\alpha^{k+1} + (1-p_{r\ell,m})(2^k-1)(2\alpha)^k(N-2^k) + \epsilon_{r\ell,m}\alpha m \\
&\qquad + p_{r\ell,m}\frac{1-(1-p_{\ell r,m}-\epsilon_{\ell r,m})^m}{p_{\ell r,m}+\epsilon_{\ell r,m}} \left(p_{\ell r,m}\alpha^{k+1} + (1-p_{\ell r,m})(2^k-1)(2\alpha)^k(N-2^k) + \epsilon_{\ell r,m}\alpha m\right).
\end{align*}
Due to \eqref{eq:LbMemorylessAss}, it follows that $\epsilon_{\ell r}=\epsilon_{r\ell}=0$. Thus, letting $m\to\infty$ and using that $p_{\ell r}>0$, we get
\begin{align}
    &\limsup_{m\to\infty}E(\cost_{A_{k+1}}(C_{r\ell}\cup\{r\},\sigma_r)\nonumber\\
    &\ge p_{r\ell}\alpha^{k+1} + (1-p_{r\ell})(2^k-1)(2\alpha)^k(N-2^k) + \frac{p_{r\ell}}{p_{\ell r}} \left(p_{\ell r}\alpha^{k+1} + (1-p_{\ell r})(2^k-1)(2\alpha)^k(N-2^k)\right)\nonumber\\
    &= 2p_{r\ell}\alpha^{k+1} + \left(1-2p_{r\ell} + \frac{p_{r\ell}}{p_{\ell r}}\right)(2^k-1)(2\alpha)^k(N-2^k).\label{eq:LbMemorylessSigmar}
\end{align}

The next claim says that with arbitrarily high probability, the algorithm returns to its initial configuration after each subsequence $\sigma_r$.
\begin{claim}
For all $i\le \alpha$, $P(F_r(\sigma_r^i)=C_{r\ell}\cup\{r\})\to 1$ as $m\to\infty$.\label{cl:LbMemorylessReturn}
\end{claim}
\begin{proof}
    Since
    \begin{align*}
        P(F_r(\sigma_r^i)=C_{r\ell}\cup\{r\})\ge P(\forall j=1,\dots,i\colon F_r(\sigma_r^j)=C_{r\ell}\cup\{r\}) = P(F_r(\sigma_r)=C_{r\ell}\cup\{r\})^i,
    \end{align*}
    we only need to show the claim for $i=1$. Thanks to \eqref{eq:LbMemorylessp}, \eqref{eq:LbMemorylesspeps} and $\epsilon_{r\ell}=0$, it suffices to show
    \begin{align}
        &P(F_r(\sigma_{r})=C_{r\ell}\cup\{r\}\mid F_r(\sigma_{r\ell}C_{r\ell}^m)=C_{r\ell}\cup\{\ell\}) \to 1\text{ as $m\to\infty$}\label{eq:LbMemorylessEndr1}\\
        &P(F_r(\sigma_{r})=C_{r\ell}\cup\{r\}\mid F_r(\sigma_{r\ell}C_{r\ell}^m)=C_{r\ell}\cup\{r\}) = 1. \label{eq:LbMemorylessEndr2}
    \end{align}
    
    We first show \eqref{eq:LbMemorylessEndr2}. When arriving in $C_{r\ell}\cup\{r\}$ after $\sigma_{r\ell} C_{r\ell}^m$, then the relocation sequence $(C_{r\ell}\to C_{\ell r})$ changes the configuration to $C_{\ell r}\cup\{r\}$. Thanks to part \ref{it:LbMemorylessFree} of the induction hypothesis, the following $\sigma_{\ell r}$ is then served for free and $A_{k+1}$ returns to configuration $C_{\ell r}\cup\{r\}$ at the end of it. The subsequent subsequence $C_{\ell r}^m$ does not cause any movement. Thus, at the end of the whole subsequence $(\sigma_{\ell r}C_{\ell r}^m)^m$, the configuration is $C_{\ell r}\cup\{r\}$. The last set of relocation requests changes the configuration to $C_{r\ell}\cup\{r\}$.

    For \eqref{eq:LbMemorylessEndr1}, if the configuration is $C_{r\ell}\cup\{\ell\}$ after $\sigma_{r\ell} C_{r\ell}^m$, then the relocation sequence $(C_{r\ell}\to C_{\ell r})$ changes it to $C_{\ell r}\cup\{\ell\}$. Thus,
    \begin{align*}
        &P(F_r(\sigma_{r})=C_{r\ell}\cup\{r\}\mid F_r(\sigma_{r\ell}C_{r\ell}^m)=C_{r\ell}\cup\{\ell\})\\
        &= P(F_\ell((\sigma_{\ell r}C_{\ell r}^m)^m(C_{\ell r}\to C_{r\ell}))=C_{r\ell}\cup\{r\})\\
        &\ge P(F_\ell((\sigma_{\ell r}C_{\ell r}^m)^m)=C_{\ell r}\cup\{r\})\\
        &= 1 - P(\forall i\le m\colon F_\ell((\sigma_{\ell r}C_{\ell r}^m)^i)\ne C_{\ell r}\cup\{r\})\\
        &= 1 - P(\forall i\le m\colon F_\ell((\sigma_{\ell r}C_{\ell r}^m)^i) = C_{\ell r}\cup\{\ell\}) - P(\exists i\le m\colon F_\ell((\sigma_{\ell r}C_{\ell r}^m)^i) \notin \{C_{\ell r}\cup\{\ell\},C_{\ell r}\cup\{r\}\})\\
        &= 1 - (1-p_{\ell r,m} - \epsilon_{\ell r,m})^m - P(\exists i\le m\colon C_{\ell r}\not\subset F_\ell((\sigma_{\ell r}C_{\ell r}^m)^i)),
    \end{align*}
    where the second equation uses part \ref{it:LbMemorylessFree} of the induction hypothesis in the same way as it was used in the proof of \eqref{eq:LbMemorylessEndr2}. Due to our assumptions $p_{\ell r}>0$ and \eqref{eq:LbMemorylessAss}, both subtrahends in the last term tend to $0$ as $m\to\infty$.
\end{proof}

Let $q_m=P\left(F_r(\sigma_r^\alpha(C_{r\ell}\to C_{\ell r})(C_{\ell r}\cup\{\ell\})^m)=C_{\ell r}\cup\{\ell\}\right)$. This is the probability of successfully forcing configuration $C_{\ell r}\cup\{\ell\}$, bringing $A_{k+1}$ in the symmetric situation of the initial configuration, before the ``second half'' of $\sigma$. Again by \eqref{eq:LbMemorylessAss}, we have $q_m\to 1$ as $m\to\infty$. We are now ready to put the parts together:
\begin{align*}
    &E(\cost_{A_{k+1}}(C,\sigma)) \\
    &\ge \sum_{i=0}^{\alpha-1}P(F_r(\sigma_r^i)=C_{r\ell}\cup\{r\})E(\cost_{A_{k+1}}(C_{r\ell}\cup\{r\},\sigma_r))\\
    &\quad+ q_m\sum_{i=0}^{\alpha-1}P(F_\ell(\sigma_\ell^i)=C_{\ell r}\cup\{\ell\})E(\cost_{A_{k+1}}(C_{\ell r}\cup\{\ell\},\sigma_\ell))
\end{align*}
and therefore
\begin{align*}
    &\limsup_{m\to\infty} E(\cost_{A_{k+1}}(C,\sigma))\\
    &\ge \alpha \limsup_{m\to\infty}\left[ E(\cost_{A_{k+1}}(C_{r\ell}\cup\{r\},\sigma_r))+ E(\cost_{A_{k+1}}(C_{\ell r}\cup\{\ell\},\sigma_\ell))\right]\\
    &\ge \alpha\left[
    2p_{r\ell}\alpha^{k+1} + \left(1-2p_{r\ell} + \frac{p_{r\ell}}{p_{\ell r}}\right)(2^k-1)(2\alpha)^k(N-2^k)\right.\\
    &\qquad + \left.2p_{\ell r}\alpha^{k+1} + \left(1-2p_{\ell r} + \frac{p_{\ell r}}{p_{r \ell}}\right)(2^k-1)(2\alpha)^k(N-2^k)\right]\\
    &\ge 2\alpha\left[(p_{r\ell}+p_{\ell r})\alpha^{k+1} + \left(2-(p_{r\ell}+p_{\ell r})\right)(2^k-1)(2\alpha)^k(N-2^k)\right]\\
    &\ge (2\alpha)^{k+1}\left[N - 2^k(2^k-1)(1-\frac{2^k}{N}) + (2^{k+1}-2)(N-2^k)\right]\\
    &> (2\alpha)^{k+1}(2^{k+1}-1)(N - 2^{k+1}),
\end{align*}
where the first inequality uses $q_m\to 1$ and Claim~\ref{cl:LbMemorylessReturn} and its symmetric case, the second inequality uses \ref{eq:LbMemorylessSigmar} and its symmetric case, the third inequality uses that $\frac{x}{y}+\frac{y}{x}\ge 2$ for $x,y>0$, and the fourth inequality holds for large enough $\alpha=N^2$ and uses $p_{r\ell}+p_{\ell r} > \frac{2^k}{N}$. Since the last inequality is strict, the induction step follows for large enough $m$.

\ \\
\noindent \textbf{Case 2:} $p_{\ell r}+p_{r\ell}\le \frac{2^k}{N}$ or $p_{r\ell}=0$ or $p_{\ell r}=0$.

\noindent In this case, we exploit the reluctance of $A_{k+1}(C_{r\ell}\cup\{r\},\sigma_{r\ell}C_{r\ell}^m)$ to move the taxi from $r$ to the left subtree (or likewise with $r$ and $\ell$ reversed) by requesting $\sigma_{r\ell}C_{r\ell}^m$ many times in a row. Eventually, the algorithm will have to bring the taxi from $r$ or it will suffer unbounded cost. Concretely, we define
\begin{align*}
\sigma_\ell &= (\sigma_{r\ell}C_{r\ell}^m)^m(C_{r\ell}\to C_{\ell r})\\
\sigma_r &= (\sigma_{\ell r}C_{\ell r}^m)^m(C_{\ell r}\to C_{r\ell})\\
\sigma &= (\sigma_\ell\sigma_r)^{2^kN}.
\end{align*}

We begin with the induction step of \ref{it:LbMemorylessOpt}. From initial configuration $C=C_{r\ell}\cup\{r\}$, the offline algorithm can move the taxi from $r$ to $\ell$ for cost $\alpha^{k+1}$ at the start. Now, from configuration $C_{r\ell}\cup\{\ell\}$, the sequence $(\sigma_{r\ell}C_{r\ell}^m)^m(C_{r\ell}\to C_{\ell r})$ is served for free thanks to part \ref{it:LbMemorylessFree} of the induction hypothesis, leading to configuration $C_{\ell r}\cup\{\ell\}$. Then, the taxi from $\ell$ moves back to $r$ for another cost $\alpha^{k+1}$, so that no more cost is incurred during $(\sigma_{\ell r}C_{\ell r}^m)^m(C_{\ell r}\to C_{r\ell})$, which makes the algorithm returns to $C$. Repeating this $2^kN$ times, we get $w(C,\sigma)\le 2^kN(\alpha^{k+1}+\alpha^{k+1}) \le (2\alpha)^{k+1}(N + k)$, as desired.

The proof of $0<w(C,\sigma)$ and \ref{it:LbMemorylessFree} is again straightforward.

For part \ref{it:LbMemorylessAlg}, if $p_{r\ell}=0$ or $p_{\ell r}=0$, then the cost is unbounded as $m\to\infty$ during the first subsequence $\sigma_\ell\sigma_r$ already. So we can assume $p_{r\ell}>0$ and $p_{\ell r}>0$. The same techniques as in the aggressive case yield
\begin{align*}
&E(\cost_{A_{k+1}}(C_{r\ell}\cup\{r\},\sigma_\ell))\\
&\quad\ge\sum_{i=0}^{m-1}(1-p_{r\ell,m}-\epsilon_{r\ell,m})^i E(\cost_{A_{k+1}}(C_{r\ell}\cup\{r\},\sigma_{r\ell}C_{r\ell}^m))\\
&\quad\ge \frac{1-(1-p_{r\ell,m}-\epsilon_{r\ell,m})^m}{p_{r\ell,m}+\epsilon_{r\ell,m}}\left(p_{r\ell,m}\alpha^{k+1} + (1-p_{r\ell,m})(2^k-1)(2\alpha)^k(N-2^k) + \epsilon_{r\ell,m}\alpha m\right)
\end{align*}
and, with assumption \eqref{eq:LbMemorylessAss},
\begin{align*}
P(F_r(\sigma_\ell)=C_{\ell r}\cup\{\ell\}) \to 1\text{ as $m\to\infty$}.
\end{align*}
Together with the symmetric equivalents of these statements, this gives
\begin{align*}
\limsup_{m\to\infty}E(\cost_{A_{k+1}}(\sigma,C))&\ge 2^k N\left[2\alpha^{k+1} + \left(\frac{1}{p_{r\ell}} + \frac{1}{p_{\ell r}} - 2\right)(2^k-1)(2\alpha)^k(N-2^k)\right]\\
&\ge (2\alpha)^{k+1} (N - (2^k-1)2^k) + 2^kN\frac{4}{p_{\ell r}+p_{r\ell}}(2^k-1)(2\alpha)^k(N-2^k)\\
&\ge (2\alpha)^{k+1} \left(N - (2^k-1)2^k + (2^{k+1}-2)(N-2^k)\right)\\
&> (2\alpha)^{k+1} (2^{k+1}-1)(N - 2^{k+1}),
\end{align*}
where the second inequality uses $\frac{1}{x}+\frac{1}{y}\ge \frac{4}{x+y}$ for $x,y>0$ and the third inequality holds for large enough $N$ and uses $p_{\ell r}+p_{\ell r}\le \frac{2^k}{N}$. This completes the proof.
\end{proof}

\subsection{Reduction from layered graph traversal}\label{subsec:redLGT}

The \emph{layered width-$k$ graph traversal problem ($k$-LGT)} is defined as follows: A searcher starts at a node $s$ of a graph with non-negative edge weights and whose nodes can be partitioned into layers $L_0=\{s\},L_1,L_2,\dots$ such that all edges run between consecutive layers. Each layer contains at most $k$ nodes. The goal is to move the searcher along the edges to some vertex $t$ while minimizing the distance traveled by the searcher. However, the nodes in $L_\ell$ and the edges between $L_{\ell-1}$ and $L_\ell$ are only revealed when the searcher reaches a node in $L_{\ell-1}$.

It is known that the deterministic competitive ratio of $k$-LGT is between $\Omega(2^k)$ \cite{FiatFKRRV98} and $O(k2^k)$ \cite{Burley96}, and it is $9$ for $k=2$ \cite{PapadimitriouY91}. By the following reduction, these lower bounds translate to the hard $k$-taxi problem, giving Theorem~\ref{thm:red} as well as the lower bound of Theorem~\ref{thm:k=2}.\footnote{Recall that the $\Omega(2^k)$ lower bound also follows already from the lower bound for adaptive adversaries in Theorem~\ref{thm:HSTs}.}

\begin{theorem}
	If there exists a $\rho$-competitive deterministic algorithm for the hard $k$-taxi problem, then there exists a $\rho$-competitive deterministic algorithm for $k$-LGT.
\end{theorem}
\begin{proof}
	Fiat et al.~\cite{FiatFKRRV98} showed that $k$-LGT has the same competitive ratio as its restricted case where the graph is a tree and all edges have weight $0$ or $1$. Let $s_\ell$ be the first node visited by the online algorithm in the $\ell$th layer; in particular, $s_0$ is the starting position of the searcher. We can assume that any node $v\in L_\ell\setminus\{s_\ell\}$ has at most one adjacent node in layer $\ell+1$, and the connecting edge would be of weight $0$. This is because any other edges leaving $v$ can be delayed to a later layer, once the searcher moves to that branch. We design a $\rho$-competitive algorithm for the traversal of this type of $0$-$1$-weighted trees. The movement of the searcher is determined by the decisions of a $\rho$-competitive $k$-taxi algorithm.
	
	Let $T$ be the layered tree of width at most $k$ for the traversal problem with the aforementioned properties. As metric space for the $k$-taxi algorithm we use an infinite tree where each node has infinitely many children and each edge has weight $1$. Note that $T$ can be isometrically embedded into this infinite tree by contracting any nodes connected by an edge of weight $0$ to a single node. Moreover, such an embedding can be constructed online while $T$ is being revealed. We will only make taxi requests at nodes corresponding to the revealed part of $T$, so we can pretend that the requests and taxis are located on $T$ itself.
	
	We will maintain the following invariant: Right after a layer $L_\ell$ gets revealed because the searcher moved to $s_{\ell-1}$, the configuration $C_\ell$ of the $k$ taxis is a multiset over $L_\ell$ where each node of $L_\ell$ has multiplicity at least $1$. Initially this situation can be achieved by relocation requests. Then, a simple request at $s_{\ell-1}$ is issued. The $k$-taxi algorithm will serve the request by moving from some $x\in C_\ell$ to $s_{\ell-1}$. We will move the same distance in the traversal problem by moving the searcher from $s_{\ell-1}$ to $x$, which reveals the $(\ell+1)$st layer and sets $s_\ell:=x$. Some of the taxis may be able to move to layer $\ell+1$ for free along edges of weight $0$. By making relocation requests it will be ensured that all taxis occupy all of the at most $k$ nodes in layer $\ell+1$, so the invariant holds again. This defines a procedure for traversing the tree.
	
	Note that the online cost for traversing the tree is the same as the online cost for the $k$-taxi problem. It only remains to show that the optimal cost for the traversal problem is at least the optimal cost for the $k$-taxi problem.
	
	Let $\sigma_{\ell}$ be the request sequence up to the point where the online taxis are in configuration $C_\ell$. For $y\in C_\ell$ we denote by $C_\ell-y+s_{\ell-1}$ the configuration obtained from $C_\ell$ by replacing one copy of $y$ by $s_{\ell-1}$. Let $w_\ell(C_\ell-y+s_{\ell-1})$ be the optimal cost to serve $\sigma_{\ell}$ and subsequently end up in configuration $C_\ell-y+s_{\ell-1}$. We claim that $w_\ell(C_\ell-y+s_{\ell-1})\le d(s_0,y)$ for all $y\in C_\ell$, where $d$ denotes the distance function on $T$. We prove this by induction on $\ell$.
	
	After the initial relocation requests, the offline configuration is $C_1$ and configuration $C_1-y+s_{0}$ can be reached for cost $d(s_0,y)$ by moving a taxi from $y$ to $s_0$. Thus, the claim holds for $\ell=1$.
	
	Suppose the claim holds for some $\ell$. The next requests after $\sigma_{\ell}$ are a simple request at $s_{\ell-1}$ (which changes the online configuration from $C_\ell$ to $C_\ell-s_\ell+s_{\ell-1}$) and some relocation requests that, together with moves along weight-$0$ edges, change the configuration to $C_{\ell+1}$. Let $y'\in C_{\ell+1}$ and let $y\in C_\ell$ be its parent in layer $\ell$.  One (offline) way to serve $\sigma_{\ell+1}$ and end up in configuration $C_{\ell+1}-y'+s_\ell$ is as follows: First serve $\sigma_{\ell}$ and reach configuration $C_\ell-y+s_{\ell-1}$ for cost $w_\ell(C_\ell-y+s_{\ell-1})$. The simple request at $s_{\ell-1}$ is then served for free without moving. Recall that the following relocation requests and moves along weight-$0$ edges change the online configuration from $C_\ell-s_\ell+s_{\ell-1}$ to $C_{\ell+1}$. If the edge $(y,y')$ has weight $0$, then the same relocations and moves along weight-$0$ edges, except for the move from $y$ to $y'$, change the offline configuration from $C_\ell-y+s_{\ell-1}$ to $C_{\ell+1}-y'+s_{\ell}$. So in this case,
	\begin{align*}
	w_{\ell+1}(C_{\ell+1}-y'+s_\ell) \le w_\ell(C_\ell-y+s_{\ell-1}) \le d(s_0,y)=d(s_0,y')
	\end{align*}
	as claimed. In the other case, $y=s_\ell$ by assumption on the layered tree, so before the relocation requests, the offline configuration is the same as the online configuration. Thus, online and offline configuration are the same also after the relocation moves, which is configuration $C_{\ell+1}$. Finally, $C_{\ell+1}-y'+s_\ell$ can be reached by moving a taxi from $y'$ to $y=s_\ell$. Thus,
	\begin{align*}
	w_{\ell+1}(C_{\ell+1}-y'+s_\ell) \le w_\ell(C_\ell-y+s_{\ell-1}) + d(y,y') \le d(s_0,y) + d(y,y')=d(s_0,y').
	\end{align*}
	
	From the claim it follows that the optimal cost for the taxi request sequence is at most the length of the path from $s_0$ to any node $y$ in the last layer. If $y$ is the target vertex of the traversal problem, the latter is precisely the offline cost for the traversal problem.
\end{proof}

\subsection{An optimally competitive algorithm for two taxis}\label{subsec:k=2}
We will define a deterministic algorithm \textsc{BiasedDC} for the hard $2$-taxi problem on general metrics.

Note that there is always a pair of an online algorithm taxi and an offline algorithm taxi occupying the same location, namely the taxis that served the last request (or, in the initial configuration, any online taxi and the corresponding offline taxi starting at the same point). We call these taxis \emph{active} and denote them by $A$ (online) and $a$ (offline). The other two taxis are \emph{passive}, denoted by $P$ (online) and $p$ (offline).

\textsc{BiasedDC} is a speed-adjusted variant of the well-known double coverage (DC) algorithm for the $k$-server problem \cite{ChrobakKPV91,ChrobakL91}. Upon a simple request at $s$, \textsc{BiasedDC} moves both taxis towards $s$, but $P$ moves at twice the speed of $A$. As soon as either taxi reaches $s$, both taxis stop moving.

This definition assumes that all points along shortest paths from the old taxi positions to $s$ belong to the metric space, which does not have to be true in general. However, we can assume that this is the case by adding virtual points to the metric space: If a taxi moves from its old position $\ell$ towards the request $s$ but stops after a fraction $q$ of the movement, we augment the metric space by adding a new point at distance $qd(\ell,s)$ from $\ell$ and $(1-q)d(\ell,s)$ from $s$, and other distances as induced by the shortest path through $s$ or $\ell$. When a taxi wants to stop at a virtual point before reaching the request, we actually leave this taxi at its old position, but when computing future moves we pretend it is located at the virtual point. By the triangle inequality, this does not increase the overall distance traveled.

The intuition is that \textsc{BiasedDC} seeks to be in a configuration similar to the offline algorithm. Before the request, $A$ was already at the position of the offline taxi $a$, whereas $P$ may have been placed suboptimally away from any offline taxi. Therefore, we prefer to move $P$ away from its old location as opposed to $A$. Accordingly, \textsc{BiasedDC} moves $P$ faster towards the request (= the new position of some offline taxi).

By the following theorem, \textsc{BiasedDC} achieves the optimal competitive ratio of $9$, matching the aforementioned lower bound and together yielding Theorem~\ref{thm:k=2}.

\begin{theorem}\label{thm:2taxiDet}
	\textsc{BiasedDC} is $9$-competitive for the hard $2$-taxi problem.
\end{theorem}
\vShort{\begin{proof}[Proof idea]
The proof is based on a standard potential function argument. The only essential difficulty is to come up with the right potential function, which in this case is three times the value of a minimum matching of the online and offline configurations.
\end{proof}}
\vFull{\begin{proof}
	We use the potential $\Phi=3M$, where $M$ is the minimum matching of the two online taxis with the two offline taxis. After serving a request, when $A$ and $a$ are both located at the same point, $M$ is simply the distance $d(p,P)$ between the two passive taxis.
	
	Let $\cost$ and $\OPT$ denote the cost of \textsc{BiasedDC} and the offline algorithm respectively for a given request, and let $\Delta\Phi$ denote the change in potential due to serving this request. It suffices to show that
	\begin{align}\label{eq:2taxiPotIneq}
	\cost+\Delta\Phi\le9\OPT.
	\end{align}
	Summing this inequality over all request yields the result because $\Phi$ is initially $0$ and remains non-negative.
	
	For relocation requests, no cost is incurred and the potential remains unchanged, hence \eqref{eq:2taxiPotIneq} is satisfied.
	
	Consider now some simple request. We can assume without loss of generality that serving the request lasts exactly one time unit, so $A$ moves distance $1$ and $P$ moves distance $2$. Thus, $\cost=3$. We distinguish two cases depending on whether $a$ or $p$ serves the request.
	
	If $a$ serves the request, then $\OPT\ge1$ because $a$ starts its movement from the same location as $A$ and moves at least as far. In the old minimum matching, $a$ was matched to $A$ and $p$ to $P$. The distance between $a$ and $A$ increased by $\OPT-1$ and the distance between $p$ and $P$ increased by at most $2$. Thus, the minimum matching increased by at most $\OPT+1$. Putting it all together, we get
	\begin{align*}
	\cost+\Delta\Phi \le 3+3(\OPT+1) \le 9\OPT\,,
	\end{align*}
	so \eqref{eq:2taxiPotIneq} is shown.
	
	If $p$ serves the request, we divide the analysis into two steps, where first the offline algorithm moves and then \textsc{BiasedDC} moves. The matching may increase by at most $\OPT$ in the first step due to the movement of $p$. During the second step, $A$ moves a distance $1$ away from its matching partner $a$, but $P$ moves a distance $2$ towards its matching partner $p$. If the matching partners change afterwards, then this would only further reduce the matching, so in the second step, the matching decreases by at least $1$. Overall, the matching increases by at most $\OPT-1$ for this request. Hence,
	\begin{align*}
	\cost+\Delta\Phi \le 3+3(\OPT-1) =3\OPT
	\end{align*}
	and \eqref{eq:2taxiPotIneq} follows again.
\end{proof}}

\subsection{A competitive algorithm for three taxis on the line}\label{sec:Line}

In this section, we present an algorithm that achieves a constant competitive ratio for the hard $3$-taxi problem when the metric space is the real line. The algorithm, which we call \textsc{RegionTracker}, is somewhat similar to \textsc{BiasedDC} in that it moves the taxis at different speeds towards the request. However, besides the location of the active taxi, the algorithm also maintains an interval around each taxi. Intuitively, the intervals are supposed to indicate regions that the taxis should explore more aggressively. Algorithm~\ref{alg:3taxi} contains the pseudocode of \textsc{RegionTracker}. An example of the steps involved in serving a simple request is depicted in Figure~\ref{fig:lineExample}.

\begin{algorithm}
    \caption{\textsc{RegionTracker}}
    \label{alg:3taxi}
    \begin{algorithmic}[1]
    \Require Initial taxi locations $x_1\le x_2\le x_3$
    \State $A\gets 1$
    \State $(r_0,\ell_1,r_1,\ell_2,r_2,\ell_3,r_3,\ell_4)\gets (-\infty,-\infty,x_1,x_1,x_2,x_3,\infty,\infty)$
    \ForEach{request $(s,t)$}
        \If{$s<x_2$}\label{lin:SimpleStart}
            \While{$s\notin\{x_1,x_2\}$}\label{lin:WhileStart}
                \State Change $x_1,x_2,x_3$ at rates specified in Table~\ref{tab:3taxiRates}
                \State $r_1\gets (x_1\lor r_1)\land x_2$\label{lin:UpdateFrontierStart}
                \State $\ell_2\gets r_1\lor (\ell_2\land x_2)$\label{lin:UpdateFrontierEnd}
            \EndWhile\label{lin:WhileEnd}
            \State $A\gets \min\{i\mid x_i=s\}$
            \While{$\ell_A<x_A<r_A$}\label{lin:ShrinkStart}
                \State Increase $\ell_A$ and decrease $r_A$ at the same rate
            \EndWhile \label{lin:ShrinkEnd}
            \While{$\ell_A<x_A<\ell_{A+1}$}\label{lin:ShiftStart}
                \State Increase $\ell_A$ and decrease $\ell_{A+1}$ at the same rate
            \EndWhile
            \While{$r_{A-1}<x_A<r_{A}$}
                \State Decrease $r_A$ and increase $r_{A-1}$ at the same rate
            \EndWhile\label{lin:ShiftEnd}
        \ElsIf{$s>x_2$}
            \State Act symmetrically to case ``$s<x_2$''
        \EndIf\label{lin:SimpleEnd}
        \State $(e_1,e_2) \gets (r_1,\ell_2,r_2,\ell_3)\setminus(x_A,x_A)$\label{lin:e1e2}
        \State $x_A\gets t$\label{lin:RelocStart}
        \State $(x_1,x_2,x_3)\gets\text{sort}(x_1,x_2,x_3)$
        \State $A\gets \min\{i\mid x_i=t\}$
        \State $(r_1,\ell_2,r_2,\ell_3)\gets \text{sort}(e_1,e_2,x_A,x_A)$ \label{lin:RelocEnd}
    \EndFor
    \end{algorithmic}
\end{algorithm}

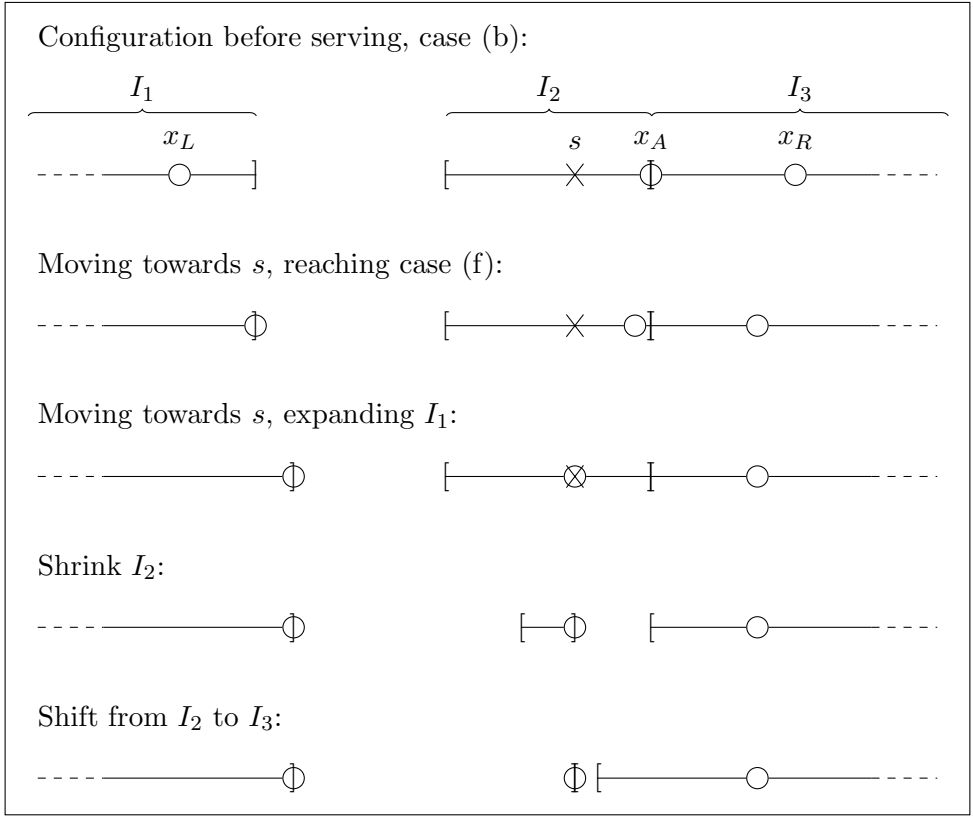
\begin{figure}\label{fig:lineExample}
\centering
	\begin{tikzpicture}[framed,scale=1,
    server/.style={draw,shape=circle,inner sep=0pt,minimum size=8pt},
    request/.style={draw,shape=cross out,inner sep=0pt,minimum width=6pt,minimum height=8pt},
    int/.style={decoration={brace,raise=8pt},decorate}]

	\newcommand{\Y}{0}
	\node[anchor=west,text width=280pt] at (-2,\Y+1.8) {Configuration before serving, case (b):};
	
	\node (-inf) at (-2,\Y) {};
	\node[server] (x1) at (0,\Y) {};
	\node (r1) at (1.0,\Y) {$]$};
	\node (l2) at (3.5,\Y) {$[$};
	\node[request] (s) at (5.2,\Y) {};
	\node[server] (x2) at (6.2,\Y) {};
	\node (r2) at (6.2,\Y) {$]$};
	\node (l3) at (6.2,\Y) {$[$};
	\node[server] (x3) at (8.1,\Y) {};
	\node (inf) at (10.1,\Y) {};
	
	\node at (x1) [above=6pt] {$x_L$};
	\node at (s) [above=6pt] {$s$};
	\node at (x2) [above=6pt] {$x_A$};
	\node at (x3) [above=6pt] {$x_R$};
	
	\draw[dashed] (-inf) -- (-1,\Y);
	\draw (-1,\Y) -- (x1) -- ($ (r1) + (0.01,0) $);
	\draw ($ (l2) + (-0.01,0) $) -- (x2) -- (x3) -- (9.1,\Y);
	\draw[dashed] (9.1,\Y) -- (inf);
		
	\draw[int] ($ (-inf) + (0.01,0.5) $) -- node (I1) [above=10pt] {$I_1$}  ($ (r1) + (0.01,0.5) $);
	\draw[int] ($ (l2) + (0.01,0.5) $) -- node (I2) [above=10pt] {$I_2$}  ($ (r2) + (0.01,0.5) $);
	\draw[int] ($ (l3) + (0.01,0.5) $) -- node (I3) [above=10pt] {$I_3$}  ($ (inf) + (0.01,0.5) $);

	\renewcommand{\Y}{-2}
	\node[anchor=west,text width=280pt] at (-2,\Y+0.8) {Moving towards $s$, reaching case (f):};
	
	\node (-inf) at (-2,\Y) {};
	\node[server] (x1) at (1.0,\Y) {};
	\node (r1) at (1.0,\Y) {$]$};
	\node (l2) at (3.5,\Y) {$[$};
	\node[request] (s) at (5.2,\Y) {};
	\node[server] (x2) at (5.99,\Y) {};
	\node (r2) at (6.2,\Y) {$]$};
	\node (l3) at (6.2,\Y) {$[$};
	\node[server] (x3) at (7.6,\Y) {};
	\node (inf) at (10.1,\Y) {};
	
	\draw[dashed] (-inf) -- (-1,\Y);
	\draw (-1,\Y) -- (x1);
	\draw ($ (l2) + (-0.01,0) $) -- (x2) -- (x3) -- (9.1,\Y);
	\draw[dashed] (9.1,\Y) -- (inf);

	\renewcommand{\Y}{-4}
	\node[anchor=west,text width=280pt] at (-2,\Y+0.8) {Moving towards $s$, expanding $I_1$:};
	
	\node (-inf) at (-2,\Y) {};
	\node[server] (x1) at (1.5,\Y) {};
	\node (r1) at (1.5,\Y) {$]$};
	\node (l2) at (3.5,\Y) {$[$};
	\node[request] (s) at (5.2,\Y) {};
	\node[server] (x2) at (5.2,\Y) {};
	\node (r2) at (6.2,\Y) {$]$};
	\node (l3) at (6.2,\Y) {$[$};
	\node[server] (x3) at (7.6,\Y) {};
	\node (inf) at (10.1,\Y) {};
	
	\draw[dashed] (-inf) -- (-1,\Y);
	\draw (-1,\Y) -- (x1);
	\draw ($ (l2) + (-0.01,0) $) -- (x2) -- (x3) -- (9.1,\Y);
	\draw[dashed] (9.1,\Y) -- (inf);

	\renewcommand{\Y}{-6}
	\node[anchor=west,text width=280pt] at (-2,\Y+0.8) {Shrink $I_2$:};
	
	\node (-inf) at (-2,\Y) {};
	\node[server] (x1) at (1.5,\Y) {};
	\node (r1) at (1.5,\Y) {$]$};
	\node (l2) at (4.5,\Y) {$[$};
	\node[server] (x2) at (5.2,\Y) {};
	\node (r2) at (5.2,\Y) {$]$};
	\node (l3) at (6.2,\Y) {$[$};
	\node[server] (x3) at (7.6,\Y) {};
	\node (inf) at (10.1,\Y) {};
	
	\draw[dashed] (-inf) -- (-1,\Y);
	\draw (-1,\Y) -- (x1);
	\draw ($ (l2) + (-0.01,0) $) -- (x2);
	\draw ($ (l3) + (-0.01,0) $) -- (x3) -- (9.1,\Y);
	\draw[dashed] (9.1,\Y) -- (inf);

	\renewcommand{\Y}{-8}
	\node[anchor=west,text width=280pt] at (-2,\Y+0.8) {Shift from $I_2$ to $I_3$:};
	
	\node (-inf) at (-2,\Y) {};
	\node[server] (x1) at (1.5,\Y) {};
	\node (r1) at (1.5,\Y) {$]$};
	\node (l2) at (5.2,\Y) {$[$};
	\node[server] (x2) at (5.2,\Y) {};
	\node (r2) at (5.2,\Y) {$]$};
	\node (l3) at (5.5,\Y) {$[$};
	\node[server] (x3) at (7.6,\Y) {};
	\node (inf) at (10.1,\Y) {};
	
	\draw[dashed] (-inf) -- (-1,\Y);
	\draw (-1,\Y) -- (x1);
	\draw ($ (l3) + (-0.01,0) $) -- (x3) -- (9.1,\Y);
	\draw[dashed] (9.1,\Y) -- (inf);
	
	\end{tikzpicture}
\caption{Example of \textsc{RegionTracker} serving a simple request at $s$.}
\end{figure}

At any point in time, we denote by $x_1\le x_2\le x_3$ the locations of the algorithm's taxis. The index $A\in\{1,2,3\}$ indicates the \emph{active taxi} that served the last request (or $A=1$ initially). We use variables $r_1\le \ell_2\le r_2\le \ell_3$ to represent the intervals $I_1=(-\infty,r_1]$, $I_2=[\ell_2,r_2]$ and $I_3=[\ell_3,\infty)$, and we will ensure at all times that $x_i\in I_i$ for $i=1,2,3$. For technical reasons, we also define $r_0=\ell_1=-\infty$ and $r_3=\ell_4=\infty$.

Before and after serving a request, it will always be the case that two of the four finite interval endpoints are equal to the location $x_A$ of the active taxi. We sometimes denote the other two interval endpoints by $e_1\le e_2$, and the two passive taxis by $L=\min\{1,2,3\}\setminus\{A\}$ and $R=\max\{1,2,3\}\setminus\{A\}$. Let $\emph{sort}$ be the operator that maps a sequence of numbers to the same sequence sorted in non-decreasing order.

\begin{observation}\label{obs:e1e2}
If $(r_1,\ell_2,r_2,\ell_3)=\text{sort}(e_1,e_2,x_A,x_A)$, then $x_L\in(-\infty,e_1]$ and $x_R\in[e_2,\infty)$.
\end{observation}

Given a taxi request $(s,t)$, \textsc{RegionTracker} moves a taxi to $s$ as follows. For the sake of this description, let us assume that $s<x_2$; the other case is symmetric. If $s\le x_1$, then we simply move the leftmost taxi to $s$. Otherwise, in most cases (see Table~\ref{tab:3taxiRates}) we move the two adjacent taxis continuously towards $s$ until one of them reaches $s$. If one of them has reached the frontier of its interval and the other one has not, then the one which is still in the interior of its interval moves by a factor $b+1$ or $c+1$ faster than the one that is already at the frontier, for constants $c>b>0$. However, there is one exception: If none of the three taxis has reached the interval frontier between itself and $s$, then all three taxis move towards the request at speeds $b+1$, $1$ and $b$. Simultaneously to moving the taxis, we will also update the interval frontiers so as to ensure that $x_i\in I_i$ continues to hold and the interiors of $I_1$, $I_2$ and $I_3$ are disjoint (lines~\ref{lin:UpdateFrontierStart}--\ref{lin:UpdateFrontierEnd} of Algorithm~\ref{alg:3taxi}). We use here and onwards the notation $x\land y=\min\{x,y\}$ and $x\lor y=\max\{x,y\}$. In other words, if a taxi reaches its interval frontier, then it pushes this frontier further as it is moving; if it reaches the interval frontier of an adjacent taxi, it pushes the frontier back. Once $s$ is reached, the active taxi index $A$ is updated.

Since pushing the frontiers and updating $A$ may have violated the property that the list $r_1,\ell_2,r_2,\ell_3$ contains two copies of $x_A$, we need to do some post-processing. In lines~\ref{lin:ShrinkStart}--\ref{lin:ShrinkEnd}, we shrink interval $I_A$ until one of the endpoints reaches $x_A$. Thereafter (lines~\ref{lin:ShiftStart}--\ref{lin:ShiftEnd}), we ``shift'' any remaining part of $I_A$ to the side. More precisely, if $\ell_A<x_A=r_A$ (which can only be the case for $A=1,2$), then we push the frontier $\ell_A$ towards $x_A$ (further shrinking $I_A$) while pulling $\ell_{A+1}$ away from $x_{A+1}$ towards $x_A$ (enlarging $I_{A+1}$). This is done until either $\ell_A$ or $\ell_{A+1}$ reaches $x_A$. We act similarly if instead we had $\ell_A=x_A<r_A$ after line~\ref{lin:ShrinkEnd}. After this, it is indeed true that (at least) two of the interval endpoints $r_1,\ell_2,r_2,\ell_3$ are equal to $x_A$.

In line~\ref{lin:e1e2}, we define $e_1\le e_2$ as the other two interval endpoints, as mentioned above. (In the pseudocode, we use the set difference notation to remove elements from a list.) To serve the relocation part of the request, we simply change the location of $x_A$, make sure that $x_1\le x_2\le x_3$ are again in the right order, and update $A$ accordingly. Finally, we update the interval endpoints to react to the relocation.

\newcounter{rowcount}
\newcommand\rownumber{\stepcounter{rowcount}(\alph{rowcount})}
\begin{table}[h]
    \centering
    \begin{tabular}{cl|c|c|c||c|c}
        &\text{Conditions} & $x_1'$ & $x_2'$ & $x_3'$ & $\Sigma'\le$ & if $a$ changed: $\Psi'\le$ \\[1pt] \hline
        \rownumber& $s<x_1$ & $-1$ & $0$ & $0$ & ``$0$'' & $-(\gamma-\psi)$\\
        \rownumber&$s>x_1, x_1<r_1, \ell_2<x_2, \ell_3<x_3$ & $b+1$ & $-1$ & $-b$ & $-b$ & $0$\\
        \rownumber&$s>x_1, x_1<r_1, \ell_2<x_2, \ell_3=x_3$ & $1$ & $-1$ & $0$ & $-1$ & $0$ \\
        \rownumber&$s>x_1, x_1=r_1, \ell_2=x_2$ & $1$ & $-1$ & $0$ & $1$ & $-\psi$ \\
        \rownumber&$s>x_1, x_1<r_1, \ell_2=x_2$ & $b+1$ & $-1$ & $0$ & $1$ & $-(\gamma-\psi)b$\\
        \rownumber&$s>x_1, x_1=r_1, \ell_2<x_2, A\ge 2$ 
        & $1$ & $-(b+1)$ & $0$ & $b+1$ & $-(\gamma-\psi)b$\\
        \rownumber&$s>x_1, x_1=r_1, \ell_2<x_2, A=1$ & $1$ & $-(c+1)$ & $0$ & $c+1$ & $2\psi-c(\gamma-\psi)$
    \end{tabular}
    \caption{Rates of movement if $s\in(-\infty,x_1)\cup(x_1,x_2)$, where $c>b>0$ are constants. Last two columns for analysis.}
    \label{tab:3taxiRates}
\end{table}

\subsubsection*{Analysis} We use the notation $\inte{x,y}=[x\land y,x\lor y]$ for the interval between $x$ and $y$. For the locations of the offline taxis, we write $y_1\le y_2\le y_3$, and $a\in\{1,2,3\}$ for the index of the active offline taxi.

To prove that \textsc{RegionTracker} is $O(1)$-competitive, we use a potential consisting of two parts. One of them is
\begin{align*}
    \Psi = \int_{-\infty}^\infty (w_1(z)+w_2(z)+w_3(z))dz
\end{align*}
where
\begin{align*}
w_i(z) = \begin{cases}
0, &\text{ if }z\notin\inte{x_i,y_i},\\
\gamma-\psi, &\text{ if }z\in\inte{x_i,y_i}\cap I_i,\\
\gamma+\psi, &\text{ if }z\in\inte{x_i,y_i}\setminus[r_{i-1},\ell_{i+1}],\\
\gamma, &\text{ otherwise.}
\end{cases}
\end{align*}
for some constants $\psi>0$ and $\gamma=\frac{2b+1}{2b}\psi$. We can think of $\Psi$ as a special weighted matching of the online and offline configurations: The $i$th online taxi is matched to the $i$th offline taxi. The interval between them is partitioned into (up to) three segments whose contribution to $\Psi$ is their length weighted by some factor. The segment that is in $I_i$ has weight $\gamma-\psi$, the (possible) segment between the frontiers of $I_i$ and the adjacent taxi's $I_j$ has weight $\gamma$ and a possibly remaining segment from the boundary of the adjacent taxi's $I_j$ to the offline taxi has weight $\gamma+\psi$.

The other part of the potential is
\begin{align*}
    \Sigma=\begin{cases}
        (r_1-x_1)\land (x_2-\ell_2), &\text{ if $\ell_3=x_3$,}\\
        (r_2-x_2)\land (x_3-\ell_3), &\text{ if $x_1=r_1$,}\\
        (r_1-x_1+r_2-x_2)\land (x_2-\ell_2+x_3-\ell_3) &\text{ otherwise.}
    \end{cases}
\end{align*}
Note that $\Sigma$ is well-defined, since if $x_1=r_1$ and $\ell_3=x_3$, then $\Sigma$ equates to $0$ by both the first or the second case of the definition. This part has a purpose somewhat similar to the ``sum of pairwise server distances'' part in the potential of the Double Coverage algorithm for the $k$-server problem \cite{ChrobakKPV91}. For the hard $k$-taxi problem, a plain pairwise server distances potential does not make sense since these distances can be changed arbitrarily by relocation requests; instead, $\Sigma$ is a variant of this that represents the distance between the two passive online taxis, truncated at the closest $e_i$:

\begin{claim}\label{cl:3taxiSigma}
If $(r_1,\ell_2,r_2,\ell_3) = \text{sort}(e_1,e_2,x_A,x_A)$, then $\Sigma=(e_1-x_L)\land(x_R-e_2)$.
\end{claim}
\begin{proof}
By Observation~\ref{obs:e1e2}, we have $x_L\le e_1\le e_2\le x_R$.

If $A=1$, then $(r_1,\ell_2,r_2,\ell_3)=(x_1,x_1,e_1,e_2)$. Therefore $\Sigma=(r_2-x_2)\land (x_3-\ell_3) = (e_1-x_L)\land (x_R-e_2)$. The case $A=3$ is similar.

For $A=2$, we consider several subcases. If $x_1=e_1$, then $r_1=x_1$ and $r_2=x_2$. Therefore $\Sigma=(r_2-x_2)\land (x_3-\ell_3)= 0 = (e_1-x_L)\land(x_R-e_2)$. The same argument handles the case $x_3=e_2$, so let us assume $x_1<e_1$ and $e_2<x_3$. If $e_1\le x_2\le e_2$, then $(r_1,\ell_2,r_2,\ell_3)=(e_1,x_2,x_2,e_2)$ and $\Sigma=(r_1-x_1+r_2-x_2)\land (x_2-\ell_2+x_3-\ell_3) = (e_1-x_L)\land (x_R-e_2)$. If $x_2<e_1$, then $(r_1,\ell_2,r_2,\ell_3)=(x_2,x_2,e_1,e_2)$. If $x_1=x_2$, then $\Sigma=(r_2-x_2)\land (x_3-\ell_3)=(e_1-x_L)\land (x_R-e_2)$. If $x_1<x_2$, then $\Sigma=(r_1-x_1+r_2-x_2)\land (x_2-\ell_2+x_3-\ell_3) = (e_1-x_L)\land (x_R-e_2)$. The case $x_2>e_2$ is symmetric to $x_2<e_1$.
\end{proof}

As the overall potential, we use $\Phi=\alpha\Sigma + \Psi$ for some constant $\alpha>0$.

\begin{claim}
$\Phi$ remains constant during the relocation in lines~\ref{lin:RelocStart}--\ref{lin:RelocEnd} of Algorithm~\ref{alg:3taxi}.
\end{claim}
\begin{proof}
By Claim~\ref{cl:3taxiSigma}, $\Sigma$ depends only on $x_L, e_1, e_2,x_R$, which do not change under relocation.

To show that also $\Psi$ remains unchanged, we show that the value of $w_1(z)+w_2(z)+w_3(z)$ is independent of the location $y_a=x_A$ of the active taxi pair for almost all $z$. To do so, we determine the value of $w_1(z)+w_2(z)+w_3(z)$ for $z\notin\{x_1,x_2,x_3\}$. Let $\delta_z = |\{i\colon x_i < z\}| - |\{i\colon y_i < z\}|$ be the number of taxis that the online algorithm has to the left of $z$ more than the offline algorithm. Then $\delta_z\in\{-2,-1,0,1,2\}$, and $\delta_z$ is invariant under relocation of the active taxi pair.

If $\delta_z=0$, then $w_1(z)+w_2(z)+w_3(z)=0$ since $z\notin \inte{x_i,y_i}$ for all $i$.

Otherwise, let $m=\max\{i\colon x_i < z\}$. If $\delta_z=2$, then $z\in\inte{x_i,y_i}$ if and only if $i\in\{m-1,m\}$. Hence, $w_1(z)+w_2(z)+w_3(z)=w_{m-1}(z)+w_m(z)$. Since $r_{m-2}\le x_{m-1}\le \ell_m\le x_m<z$, we have $w_{m-1}(z)=\gamma+\psi$. Moreover, if $A>m$ then $r_m=x_A>z$ and otherwise $r_m=\infty$. In either case, $z\in I_m$ and hence $w_m(z)=\gamma-\psi$. Thus, $w_1(z)+w_2(z)+w_3(z)=2\gamma$ independent of the location of the active taxis.

If $\delta_z=1$, then $w_1(z)+w_2(z)+w_3(z)=w_m(z)$. We consider several sub-cases. If $x_R<z$, then $w_m(z)=\gamma-\psi$ as in the previous case. Otherwise, $x_L<z<x_R$. If $z\le e_1$, then either $x_A\in(z, e_1]$ and $r_m=x_A$ or $x_A\notin(z,e_1]$ and $r_m=e_1$. In both cases, $z\le r_m$ and therefore $w_m(z)=\gamma-\psi$.

If $e_2< z$, then either $x_A\in[e_2,z)$ and $\ell_{m+1}=x_A<z$ or $x_A\notin[e_2,z)$ and $\ell_{m+1}=e_2< z$. In both cases, $w_m(z)=\gamma+\psi$.

If $z\in (e_1,e_2]$, then either $x_A\in[e_1,z)$ and $r_m=x_A<z$ or $x_A\notin [e_1,z)$ and $r_m=e_1<z$. In both cases, $z\notin I_m$. Moreover, either $x_A\in(z,e_2]$ and $\ell_{m+1}=x_A>z$ or $x_A\notin (z,e_2]$ and $\ell_{m+1}=e_2\ge z$. In both cases, $z\in[x_m,\ell_{m+1}]\subseteq [r_{m-1},\ell_{m+1}]$. Thus, $w_m(z)=\gamma$ in this case, independent of the location of the active taxis.

The cases $\delta_z\in\{-2,-1\}$ are symmetric to $\delta_z\in\{1,2\}$.
\end{proof}

We need to show that when serving simple requests (lines~\ref{lin:SimpleStart}--\ref{lin:SimpleEnd}), the cost of \textsc{RegionTracker} plus the change of $\Phi$ is bounded by a constant times the offline cost. We first observe that $\Phi$ is non-increasing during the shrink and shift steps of the algorithm.

\begin{claim}\label{cl:Shrink}
During the shrink step (lines~\ref{lin:ShrinkStart}--\ref{lin:ShrinkEnd}), $\Phi$ does not increase.
\end{claim}
\begin{proof}
It is easy to see that $\Sigma$ does not increase.

Regarding $\Psi$, note the offline algorithm must have a taxi $y_a$ at $x_A$. If $A=a$, then clearly $\Psi$ can only decrease. If $a<A$, then $\int w_A(z)dz$ may increase at rate at most $\psi$, but at the same time $\int w_{A-1}(z)dz$ decreases at rate $\psi$. Similarly for $A>a$.
\end{proof}

\begin{claim}
During the shift step (lines~\ref{lin:ShiftStart}--\ref{lin:ShiftEnd}), $\Phi$ does not increase.
\end{claim}
\begin{proof}
Similar to the proof of Claim~\ref{cl:Shrink}.
\end{proof}

When the offline algorithm moves a taxi, $\Phi$ can only increase by at most $\gamma+\psi=O(1)$ times the distance moved by the offline algorithm. Moreover, if the offline algorithm serves the new (simple) request by moving the active taxi from $y_a$ that also served the last request, then also the cost of \textsc{RegionTracker} for this request is at most a constant times the the offline cost: This is because \textsc{RegionTracker} also has a taxi starting at $y_a=x_A$, and clearly the cost of \textsc{RegionTracker} to serve a request is at most a constant (depending on $b$ and $c$) times the distance from $x_A$ to $s$. So if the offline algorithm moves the same active taxi twice in a row, then the increase in potential plus the online cost is at most a constant times the offline cost for this request. Thus, it only remains to show now that if the offline algorithm has already moved a taxi to the new request, but this was \emph{not} the previously active offline taxi, then the cost of \textsc{RegionTracker} is cancelled by a decrease in potential. This is established in the following last Claim of this section.

\begin{claim}
If $x_A=y_a$ and $s=y_i$ for some $i\ne a$ before \textsc{RegionTracker} serves a simple request at $s$, then the movement cost of \textsc{RegionTracker} to serve the request is at most the amount by which $\Phi$ decreases at the same time.
\end{claim}
\begin{proof}
We show for all cases (a)--(g) of Table~\ref{tab:3taxiRates} that $\cost'+\Phi'\le 0$ almost always, where $\cost'=|x_1'|+|x_2'|+|x_3'|$ is the instantaneous movement cost of \textsc{RegionTracker} and $\Phi'$ the rate of change of $\Phi$. Technically, the values $x_i'$ and $\Phi'$ are only well-defined when $x_i$ and $\Phi$ are differentiable as a function of time, which they are not e.g. when the condition in Table~\ref{tab:3taxiRates} changes. However, they are differentiable almost everywhere and it suffices to show it for these times.

In case (a), we have $\cost' = 1$. Moreover, $x_1$ moves towards $y_1$ and therefore $\Psi' = -(\gamma-\psi)$. Even though $\Sigma$ may increase when $x_1$ decreases, in the subsequent shrink step $r_1$ will be reduced to the new value of $x_1$, which cancels any previous increase. So overall, $\Sigma$ does not increase. The claim follows for $\gamma-\psi$ large enough. 

In all other cases, we have $x_1<s<x_2$. Denote by $\hat{x_i}$, $\hat{\ell_i}$ and $\hat{r_i}$ the values that $x_i$, $\ell_i$ and $r_i$ had at the beginning of the while-loop. Then $y_a=\hat{x_A}$, and $(\hat{r_1},\hat{\ell_2},\hat{r_2},\hat{\ell_3})=\text{sort}(e_1,e_2,y_{a},y_a)$. Since taxis only move towards $s$, the current interval endpoints are $r_1=(\hat{r_1}\lor x_1)\land x_2$, $\ell_2=r_1\lor(\hat{\ell_2}\land x_2)$, $r_2=\hat{r_2}$ and $\ell_3=\hat{\ell_3}$.

Observe that if one of the inequalities $x_1\le r_1$, $\ell_2\le x_2$ or $\ell_3\le x_3$ becomes tight, then it remains tight throughout the while-loop. In particular, the case in the definition of $\Sigma$ changes at most once during each run of the while-loop, and $\Sigma$ can only decrease if this happens.

We will show for the cases (b) and (c) that $\Sigma$ decreases at some constant rate and $\Psi$ does not increase. Choosing $\alpha$ large enough, this will be enough to handle these cases. For the remaining cases, observe that $\Sigma$ increases at an at most constant rate (for fixed $b$ and $c$), which is immediate from the definition of $\Sigma$ and the fact that the $x_i$ change at an at most constant rate. Thus, to handle cases (d)--(g), it suffices to show that $\Psi$ decreases at an at least constant rate. Choosing $\gamma$ and $\psi$ large enough, the decrease of $\Psi$ cancels the increase of $\Sigma$ and the cost of the algorithm.

\textbf{Case (b):} We must have $\hat{r_1}=r_1$ and $\hat{\ell_2}= \ell_2$. Moreover, it must be that $y_a=\ell_3=r_2$ since otherwise the active online taxi would have moved away from $s$. The interval endpoints $\ell_i$ and $r_i$ remain constant during this case, and therefore $\Sigma' \le (-(b+1)+1)\lor (-1-b) = -b$.

For the change in $\Psi$, let us consider first the case $x_1<y_1$. Then $x_1$ moves towards $y_1$, and for each $z$ that $x_1$ moves past, $w_1(z)$ changes from $\gamma-\psi$ to $0$. So the movement of $x_1$ decreases $\Psi$ at rate $(b+1)(\gamma-\psi)$. For $i=2,3$, the movement of $x_i$ is in the worst case away from $y_i$, but it remains in the interior of $I_i$. So for any $z$ passed by $x_i$, $w_i(z)$ changes from $0$ to $\gamma-\psi$ in the worst case. So the movements of $x_2$ and $x_3$ increase $\Psi$ at rates at most $(\gamma-\psi)$ and $b(\gamma-\psi)$, respectively. Overall, $\Psi' \le 0$.

The case $y_1=x_1$ can be ignored because this will be the case only for a time interval of length $0$.

If $y_1<x_1$, then $y_2=s$ and $y_3=y_a=\ell_3$. In this case, the movement of $x_1$ increases $\Psi$ at rate $(b+1)(\gamma-\psi)$, but $x_2$ and $x_3$ move towards $y_2$ and $y_3$, respectively, decreasing $\Psi$ at rates $b(\gamma-\psi)$ and $\gamma-\psi$, respectively. Again, $\Psi'\le 0$.

\textbf{Case (c):} As in case (b), we have $y_a= \ell_3$. Thus, $y_a= x_3$.

Again, the interval endpoints remain constant during case (c). Therefore, $\Sigma'= -1$.

If $y_1=s$, then $x_1$ moves towards $y_1$, decreasing $\Psi$ at rate $\gamma-\psi$, while the movement of $x_2$ increases $\Psi$ at most at rate $\gamma-\psi$. Otherwise, $y_2=s$, the movement of $x_2$ decreases $\Psi$ at rate $\gamma-\psi$ and the movement of $x_1$ increases $\Psi$ at most at rate $\gamma-\psi$. In both cases, $\Psi'\le 0$.

As mentioned before, we show in the remaining cases only that $\Psi$ decreases at a constant rate.

\textbf{Case (d):} We have $\hat{r_1}<s<\hat{\ell_2}<y_a\le y_3$, so $s=y_1$ or $s=y_2$. If $s=y_1$, then $x_1$ moves towards $y_1$, contributing a decrease at rate $\gamma$ to $\Psi$. Even if $x_2$ moves away from $y_2$, it contributes an increase at a rate of at most $\gamma-\psi$ to $\Psi$. So in total, $\Psi'\le -\psi$. The case $s=y_2$ is similar.

\textbf{Case (e):} If $s= y_1$, then the movement of $x_1$ contributes a decrease at rate $(b+1)(\gamma-\psi)$ and the movement of $x_2$ contributes an increase at rate at most $\gamma-\psi$. Together, $\Psi'\le -b(\gamma-\psi)$. If $y_1<s$, then $y_2\le s$ and the movement of $x_1$ contributes an increase at rate at most $(b+1)(\gamma-\psi)$ while the movement of $x_2$ contributes a decrease at rate at least $\gamma$. Together, $\Psi'\le (b+1)(\gamma-\psi)-\gamma=-b(\gamma-\psi)$, since $\gamma=(2b+1)(\gamma-\psi)$.

\textbf{Case (f):} The calculations are essentially the same as in case (e).

\textbf{Case (g):} Since $A=1$, $y_a=\hat{r_1}<s$, so $s=y_2$ or $s=y_3$. For $s=y_2$ we get $\Psi'\le -c(\gamma-\psi)$ similar to cases (e) and (f). However, for $s=y_3$ it could be that $y_2\le \ell_2=r_1$, so that $x_1$'s movement is pushing $\ell_2$ towards $x_2$, leading to an additional contribution of $+2\psi$ to the change of $\Psi$. But we still have $\Psi'\le 2\psi-c(\gamma-\psi)$, which is negative for $c$ large enough.
\end{proof}

We conclude that \textsc{RegionTracker} achieves a constant competitive ratio, proving Theorem~\ref{thm:3line}.
\section{The easy \texorpdfstring{$k$}{k}-taxi problem}\label{sec:Easy}
We now turn to the easy $k$-taxi problem, and prove that it is equivalent to the $k$-server problem.

\vShort{
\begin{proof}[Proof idea of Theorem~\ref{thm:easy}]
The reduction from the easy $k$-taxi problem to the $k$-server problem is done by a transformation of a $k$-server algorithm to a $k$-taxi algorithm. The simulation is done by replacing each relocation request $(s,t)$ by a sequence of many $k$-server requests along a shortest path from $s$ to $t$. Doing it carefully, this essentially forces the same server that arrives at $s$ to move all the way to $t$. 
\end{proof}}
\vFull{
\begin{proof}[Proof of Theorem~\ref{thm:easy}]
	Clearly, since the $k$-taxi problem is a generalization of the $k$-server problem, its competitive ratio is at least that of the $k$-server problem. Thus, it suffices to show that given a $\rho$-competitive algorithm $A$ for the $k$-server problem, we can construct a $(\rho+\frac{1}{N})$-competitive algorithm $A_N$ for the easy $k$-taxi problem, for any $N\in\mathbb N$. The following proof is for deterministic algorithms. The only change that would need to be made for randomized algorithms is to replace $\cost_A$ and $\cost_{A_N}$ by their expectation.
	
	The idea of algorithm $A_N$ is to simulate the behavior of $A$ on the request sequence obtained by replacing a $k$-taxi request $(s,t)$ by many $k$-server requests along a shortest path from $s$ to $t$. In general, the underlying metric space $(M,d)$ may not contain any points on a shortest path from $s$ to $t$; we can easily fix this by embedding $M$ into a larger metric space $\tilde{M}$ that contains some additional virtual points. More precisely, $\tilde M$ is the metric space obtained by from $M$ by adding, for each $x,y\in M$, a line segment $L_{xy}$ with Euclidean metric and length $d(x,y)$ to $M$ by gluing its endpoints to $x$ and $y$ respectively. We transform a $k$-taxi request sequence $\sigma_{\textit{taxi}}$ on $M$ into a $k$-server request sequence $\sigma_{\textit{server}}$ on $\tilde{M}$ by replacing a $k$-taxi request $(s,t)$ by a subsequence $r_{0},\dots,r_{2kN}$ of $k$-server requests placed along $L_{st}$, with $r_0=s$, $r_{2kN}=t$ and distance $\frac{d(s,t)}{2kN}$ between two successive requests.
	
	Clearly,
	\begin{align*}
	\OPT(\sigma_{\textit{server}})\le \OPT(\sigma_{\textit{taxi}})
	\end{align*}
	because an optimal schedule for $\sigma_{\textit{taxi}}$ can be turned into a valid schedule for $\sigma_{\textit{server}}$ of the same cost by using the server that would serve a taxi request $(s,t)$ to serve all the associated $k$-server requests $r_0,\dots,r_{2kN}$. Therefore, since $A$ is $\rho$-competitive on $\tilde M$,
	\begin{align}
	\cost_A(\sigma_{\textit{server}}) &\le \rho \OPT(\sigma_{\textit{server}})+c \nonumber \\
	&\le \rho \OPT(\sigma_{\textit{taxi}})+c\label{eq:easyA}
	\end{align}
	for some constant $c$.
	
	The idea of algorithm $A_N$ is to transform $A$'s schedule for $\sigma_{\textit{server}}$ into a valid schedule for $\sigma_{\textit{taxi}}$ while incurring an additional cost of at most $\OPT(\sigma_{\textit{taxi}})/N$. To define $A_N$, we will pretend that taxis of $A_N$ can be located at virtual points in $\tilde M\setminus M$ even though this is not possible in the original metric space $M$. However, this will only ever happen when $A_N$ makes a move that does not serve a request, so $A_N$ does not actually have to carry out such a move and can keep the taxi in its old position until it is used to serve a request. Due to the triangle inequality, this will not increase the overall cost.
	
	We can make the following two assumptions about $A$ when it serves the subsequence $r_{0},\dots,r_{2kN}$ of equidistant requests on $L_{st}$ associated to the taxi request $(s,t)$: First, $A$ is lazy, so to serve $r_i$ it moves one server to $r_i$ and moves no other server. Second, for $i\ge2$, $A$ never serves $r_i$ with a server located at $r_j$ for some $j\le i-2$; this is because $A$ could instead move the last used server from $r_{i-1}$ to $r_i$ and (non-lazily) move the server from $r_j$ to $r_{i-1}$ to end up in the same configuration for the same cost, but then $A$ may as well delay the non-lazy move until later when/if this server is actually used to serve a request. These two assumptions mean that the requests $r_0,\dots,r_{2kN}$ can be partitioned into at most $k$ blocks of adjacent requests such that all requests within the same block are served by the same server and requests in different blocks are served by different servers. Formally, if $\ell\le k$ is the number of servers used to serve $r_0,\dots,r_{2kN}$, then there are indices $i_0=-1<i_1<\dots<i_\ell=2kN$ such that $A$ uses the $j$th of these servers to serve all the requests $r_{i_{j-1}+1},r_{i_{j-1}+2},\dots,r_{i_j}$. To turn this into a valid way to serve the taxi request $(s,t)$, we have to ensure that the same server/taxi that serves $r_0=s$ will also end up at $r_{2kN}=t$. For this, we will let the same server serve all the requests $r_0,\dots,r_{2kN}$, which can be done at a small additional cost: Namely, at the transition between blocks where $A$ uses a new server to serve $r_{i_j+1}$ instead of reusing the old server from $r_{i_j}$, algorithm $A_N$ will carry out the same server movement as $A$, followed by swapping the two servers at $r_{i_j}$ and $r_{i_j+1}$. It remains to analyze the cost of $A_N$.
	
	Since the distance between the two adjacent requests involved in a swap is $\frac{d(s,t)}{2kN}$, swapping the servers yields an additional cost of $\frac{d(s,t)}{kN}$. Therefore, total cost of all $\ell-1<k$ swaps associated with the request $(s,t)$ is at most $\frac{d(s,t)}{N}$. Over the entire request sequence, the total cost of swaps is at most a $\frac{1}{N}$ fraction of the sum of the distances of all start-destination pairs in $\sigma_{\textit{taxi}}$. Since the optimal algorithm must pay at least all of these distances, we have
	\begin{align*}
		\cost_{A_N}(\sigma_{\textit{taxi}}) &\le \cost_A(\sigma_{\textit{server}})+\frac{1}{N}\OPT(\sigma_{\textit{taxi}})\\
			&\le \left(\rho+\frac{1}{N}\right)\OPT(\sigma_{\textit{taxi}})+c,
	\end{align*}
	where the last inequality follows from \eqref{eq:easyA}.
\end{proof}}
\section{Conclusion and open problems}

The most important open problem is whether there exists an algorithm for the hard $k$-taxi problem on general metric spaces with competitive ratio based only on $k$, i.e., avoiding the dependency on $n$ in Corollary~\ref{cor:general}. We know that the Work Function Algorithm, which achieves the best known upper bound of $2k-1$ for the $k$-server problem, has unbounded competitive ratio for the hard $k$-taxi problem, even for $k=2$. However, the generalized Work Function Algorithm, a less greedy variant, may be competitive.  This algorithm is $O(k2^k)$-competitive for $k$-LGT \cite{Burley96}, but we do not see any direct way to adapt the proof to yield a similar competitive ratio for the hard $k$-taxi problem. In any case, the connection between the $k$-taxi and $k$-LGT problems is intriguing. Another way to obtain an $f(k)$-competitive algorithm for general metrics may be via dynamically updating the HST embedding, similarly to~\cite{Lee18}.

We believe that our algorithm for three taxis on the line can be the foundation to solve the problem more generally, i.e., for general $k$ and/or more general metrics such as trees or arbitrary metrics. One interesting metric space --- due to its obvious application to the $k$-taxi problem --- is the $2$-dimensional $\ell_1$-norm (also known as taxicab metric and Manhattan distance). The lower bounds of $2^k-1$ hold even for the line and the $\ell_1$-norm: This is because the \emph{binary} $\alpha$-HSTs from Theorems~\ref{thm:LbBinHST} and \ref{thm:LbMemoryless} can be embedded into the line, with a distortion tending to $1$ as $\alpha\to\infty$.

For HSTs, the main open question is whether with memory and against oblivious adversaries one can break the exponential barrier. We conjecture that the competitive ratio of $2^k-1$ on HSTs can also be achieved by a deterministic algorithm, namely the Double Coverage algorithm~\cite{ChrobakL91}. For $k=2$ this can be shown using the same potential as for \textsc{Flow} (and the fact that root-leaf-paths have the same length), however it is easy to see that this potential fails for $k>2$. For weighted star metrics one can show that Double Coverage achieves the optimal competitive ratio for the hard $k$-taxi problem, and this is $2k-1$.\footnote{For the upper bound, one can use a weighted matching as potential, where distances between online taxis and the root are scaled by a factor $2k-1$.} If our conjecture holds, then this would mean that the deterministic competitive ratio is identical to the randomized memoryless competitive ratio against oblivious adversaries on HSTs. Notice that the same is known to be true for the $k$-server problem at least on some metric spaces, where tight bounds of $k$ are known for both deterministic as well memoryless randomized algorithms (cf.~\cite{Koutsoupias09}; see also~\cite{RaghavanS94}). It would be interesting to prove this as a generic result for a broad class of online problems.

\bibliography{bibliography}{}
\bibliographystyle{plainurl}

\end{document}